\documentclass[11pt,a4paper,final]{iopart}

\usepackage{iopams}  

\usepackage[breaklinks=true,colorlinks=true,linkcolor=blue,urlcolor=blue,citecolor=blue]{hyperref}

\expandafter\let\csname equation*\endcsname\relax
\expandafter\let\csname endequation*\endcsname\relax

\usepackage{amsmath}
\usepackage[font=small,labelfont=bf,
   justification=justified,
   format=plain]{caption}
   
\usepackage{amsfonts,amssymb,mathbbol,dsfont,mathtools,pifont}
\usepackage{graphicx,psfrag,color}
\usepackage{wrapfig}
\usepackage{dcolumn,multirow}
\usepackage{bm}
\usepackage{bbm}
\usepackage{mathbbol, appendix}
\usepackage{enumerate}
\usepackage{slashbox}
\usepackage{subfigure}
\usepackage{amstext,amsthm,bm}
\usepackage[autostyle]{csquotes}
\usepackage[english]{babel}
\usepackage[T1]{fontenc}
\usepackage{float}
\usepackage[utf8]{inputenc}
\usepackage{xcolor}
\restylefloat{figure}

\newcounter{myalgctr}
\numberwithin{myalgctr}{section}

\DeclarePairedDelimiter\ceil{\lceil}{\rceil}
\DeclarePairedDelimiter\floor{\lfloor}{\rfloor}

\theoremstyle{plain}
\newtheorem{thm}{Theorem}[section]
\newtheorem{prop}[thm]{Proposition}
\newtheorem{cor}[thm]{Corollary}
\newtheorem{lemma}[thm]{Lemma}

\theoremstyle{definition}
\newtheorem{define}{Definition}

\theoremstyle{remark}
\newtheorem{remark}[thm]{Remark}

\def\pure{|\psi \rangle}
\def\dpure{|\psi \rangle \langle \psi |}

\def\N{\mathds{N}}
\def\C{\mathds{C}}
\def\H{\mathcal{H}}
\def\L{\mathcal{L}}
\def\N{\mathcal{N}}

\def \Span{{\rm span}}

\def\tr{{\rm{Tr}\,}}
\def \supp{{\rm supp}}

\def\cirtick{\makebox[0pt][l]{$\bigcirc$}\raisebox{-0.25ex}{\hspace{0.1em}$\checkmark$}}
\def\circross{\makebox[-0.5pt][l]{$\bigcirc$}\raisebox{-0ex}{\hspace{0.1em}${\bm\times}$}}

\begin{document}

\title[Generic pure states as steady states of quasi-local dissipative dynamics]
{Generic pure quantum states as steady states of quasi-local dissipative dynamics}

\author{Salini Karuvade$^1$, Peter D. Johnson$^{2}$, Francesco Ticozzi$^{3,1}$,  \\ and  Lorenza Viola$^1$}

\address{$^1$ Department of Physics and Astronomy, Dartmouth 
College, \\ 6127 Wilder Laboratory, Hanover, NH 03755, USA}

\address{$^2$ Department of Chemistry and Chemical Biology, Harvard University, \\ 
12 Oxford Street, Cambridge, MA 02138, USA}

\address{$^3$ Department of Information Engineering, University of Padua, \\ 
via Gradenigo 6/B,  35131 Padua, Italy}

\ead{lorenza.viola@dartmouth.edu}

\begin{abstract}
We investigate whether a {\em generic} multipartite pure state can be the unique asymptotic 
steady state of locality-constrained purely dissipative Markovian dynamics. In the simplest tripartite setting, 
we show that the problem is equivalent to characterizing the solution space of a set of linear equations and establish 
that the set of pure states obeying the above property has either measure zero or measure one, solely depending on 
the subsystems' dimension. A complete analytical characterization is given when the central subsystem is a qubit.
In the $N$-partite case, we provide conditions on the subsystems' size and the nature of the locality 
constraint, under which random pure states cannot be quasi-locally stabilized generically. Beside allowing for the 
possibility to {\em approximately stabilize} entangled pure states that cannot be exact steady states in settings where 
stabilizability is generic, our results offer insights into the extent to which random pure states may arise as unique 
ground states of frustration-free parent Hamiltonians. We further argue that, to high probability, pure quantum states 
sampled from a $t$-design enjoy the same stabilizability properties of Haar-random ones as long as 
suitable dimension constraints are obeyed and $t$ is sufficiently large. 
Lastly, we  demonstrate a connection between the tasks of quasi-local state stabilization and unique state reconstruction 
from local tomographic information, and provide a constructive procedure for determining a generic $N$-partite pure state 
based only on knowledge of the {\em support of any two} of the reduced density matrices of about half the parties, 
improving over existing results.
\end{abstract}

\date{\today}

\vspace{2pc}

\maketitle

\section{Introduction}

{\em Generic} states of composite quantum systems have long played an important role in the exploration of 
fundamental questions ranging from quantum information theory to the foundations of quantum and 
statistical mechanics \cite{bengtsson,rand_review}. 
For instance, consideration of {\em pure} quantum states of $N$ parties, 
drawn at random from the uniform distribution, has been instrumental in shedding light on the extent to which 
knowledge of the ``parts'' may determine the ``whole'' generically \cite{wootters_random,red_states}: that is, 
in the context of the quantum marginal (or ``local consistency'') problem \cite{YiKai,Walter}, in understanding whether a 
given set of reduced density matrices (RDMs) is compatible with a global quantum state and, if so, whether the 
latter is uniquely determined by such ``local'' information. Random pure states have also enabled 
significant progress in characterizing the nature of the correlations that a multipartite quantum state may 
exhibit generically \cite{3qubitUJ} and in developing a simplified theory of entanglement in 
high-dimensional bipartite settings, by leveraging concentration of 
measure effects \cite{Hayden2006,bengtsson}. 
In turn, similar mathematical tools have provided a fresh approach to the problem of 
thermalization and the emergence of statistical ensembles  \cite{Popescu2006,Goldstein2006}.   

Our interest in generic quantum states in this work is motivated by the problem of characterizing the 
conditions under which convergence to a stable equilibrium may be achieved in naturally occurring or 
engineered dissipative evolutions subject to realistic resource constraints.  Specifically, we ask whether 
a generic quantum pure state $\pure$ on $N$ $d$-dimensional subsystems can be 
asymptotically stable under purely dissipative Markovian dynamics subject to specified ``quasi-locality'' 
constraints, namely, whether it is \emph{dissipatively quasi-locally stabilizable} (DQLS) in the sense 
of \cite{DQLS:p,FFQLS}, as we will formalize later.
Quantum stabilization problems have been extensively investigated in recent years in the context 
of dissipative quantum state preparation and quantum engineering, with special emphasis on stabilization of 
entangled pure states of relevance to quantum information processing (QIP) \cite{nielsen} and condensed-matter physics 
-- see also \cite{Kraus2008,TV2009,VWC2009,K2011,Diehl,QLS,FTS,Kapit}
and references therein for representative contributions. 

All existing characterizations, however, have addressed stabilization of a quantum state of interest, without 
offering insight on properties that {\em stabilizable sets} may enjoy as a whole. For a fixed multipartite structure 
on the underlying Hilbert space $\H$, are there features of a target state and a quasi-local (QL) constraint 
together, that make stabilization more likely than not?  Can we possibly identify non-trivial settings where 
pure-state stabilizability is almost always feasible, so, in fact, it is typical?  
Beside adding to the conceptual understanding of the constrained stabilization problem, generic stabilizability 
would offer new venues for {\em approximate stabilization} of states that may otherwise be unattainable -- 
so-called ``practical'' stabilization in control-theoretic parlance \cite{khalil}. Additional motivation for 
characterizing the stabilizability properties of generic pure states stems from the fact that any pure state 
that enjoys the DQLS properties can be equivalently thought of as a unique ground state of a frustration-free (FF) 
QL parent Hamiltonian \cite{DQLS:p}. Thus, from a many-body physics standpoint, such a characterization 
provides insight into the existence and structure of the corresponding QL parent Hamiltonians,  
thereby complementing recent studies of ground-state properties of generic QL Hamiltonians \cite{Ramis}. 

While, throughout our analysis, we will focus for concreteness on characterizing stabilizability of a generic 
$N$-partite state $\pure$ under {\em continuous-time} QL Markovian dynamics described by a semigroup 
master equation \cite{alicki}, we stress that our approach  and conclusions are equally applicable to QL 
{\em discrete-time} dynamics, as considered in \cite{FTS,discrete}.  
The content of the paper and our main results may be summarized as follows:

$\bullet$ Sec. \ref{sec:back} introduces the relevant system-theoretic background and establishes some 
preliminary results that serve as a foundation for our subsequent analysis. In particular, after formalizing the 
QL constraints of interest in terms of a specified {\em neighborhood structure} $\N$ on $\H$, and recalling the 
existing characterization of DQLS pure states, Sec. \ref{restate} provides a reformulation of the DQLS condition 
by directly constraining the action of suitable sets of neighborhood operators acting on $\pure$ [Theorem 2.4]. 
As a by-product, this yields an improved understanding of the DQLS property itself: while the latter was known to 
be invariant under the action of local unitary transformations on the target state $\pure$ \cite{DQLS:p}, we further 
prove here invariance under the larger class of {\em Stochastic Local Operations and Classical Communication} 
(SLOCC) transformations \cite{SLOCC}. In the simplest case of a tripartite ($N=3$) setting, the reformulation of the 
DQLS condition may be further recast in the form of a system of linear equations, the solution space of which 
determines whether stabilization of $\pure$ may be achieved.

$\bullet$ Sec. \ref{sec:generictrip} is the core section of our paper, presenting a detailed analysis of the tripartite 
setting in, say, dimensions  $d_a \times d_b \times d_c$ (where $d_a \leq d_c$ without loss of generality), relative 
to two-body QL constraints, $\N_{\text{tri}} \equiv \{ \N_{ab}, \N_{bc}\}$. By utilizing the maximal Schmidt rank property 
of generic pure states along any bi-partition,  a ``no-go'' result is 
established [Theorem \ref{No-Go}], preventing QL stabilization whenever $d_a d_b \leq d_c$. 
While this immediately implies that the DQLS property is a measure zero property in such cases, our analysis makes 
it clear that, remarkably, DQLS states form either a  measure zero or a measure one set in {\em all} other possible tripartite 
settings, when $d_a d_b > d_c$ [Theorem~\ref{generic:proof}].  Sec. \ref{sub:qubit} provides a complete characterization of 
the behavior of generic tripartite pure states when $d_b=2$, essentially identifying the combination of subsystem dimensions 
$d \times 2 \times (d+1)$ as the only one for which the DQLS property holds generically [Theorem \ref{qubitb}].  
For $d_b>2$, we are still able to analytically characterize a number of settings of interest [Theorem~\ref{special_cases}], 
and arrive at a general conjecture for arbitrary $d_a,d_b, d_c$ based on numerical exploration of the remaining instances. 

$\bullet$ Sec. \ref{sec:more} collects a number of extensions and implications of the tripartite results of Sec. 
\ref{sec:generictrip}. Notably, a no-go condition for the DQLS property is also established for generic $N$-partite 
pure states [Theorem \ref{No-Go:ext}], making precise the intuition that {\em scalable} QL stabilization cannot be 
achieved generically: as $N$ increases, the set of DQLS states relative to a specified neighborhood structure $\N$ 
has measure zero if the sizes of the neighborhoods $\N_j \in \N$ (that is, physically, the range of the couplings)
are kept fixed.  When such a no-go theorem applies, we further show that no FF QL Hamiltonian 
can have a generic pure state in its ground space, as one may intuitively expect.  Building on our analysis, an 
algorithmic procedure is presented, for determining whether the DQLS property holds for given inputs $(\pure, \N)$.   
Implications beyond exact stabilization and beyond generic pure states are addressed in Sec. \ref{practical} and Sec. 
\ref{tdesign}, respectively.  In particular, we show how maximally entangled pure states may be approximately stabilized 
to arbitrary accuracy in multipartite scenarios where DQLS holds generically, and argue that a number of conclusions 
on the DQLS property (or lack thereof) may be extended to quantum states whose RDMs (or ``marginals'') do retain 
{\em certain} generic features -- notably, states sampled from a $t$-design \cite{Joseph2003,caves}.

$\bullet$ Sec. \ref{sec:marginal} illustrates connections between the task of 
pure state dissipative QL stabilization and the one of reconstructing the global state using information from its 
neighborhood RDMs, as relevant in practice to quantum state tomography protocols via local measurements 
\cite{UDAvsUDP}. Our main result, 
contained in Proposition \ref{UDA:generic}, improves over existing ones in several ways: we show that knowledge about 
the {\em support of any two RDMs on no more than (about) half} of the subsystems suffices to uniquely determine a 
generic $\pure$ on $N$ qudits among all possible quantum states;
further to that, we provide a {\em constructive} procedure for generic-state reconstruction, 
that is directly informed by the DQLS characterization we leverage in our approach.

We conclude in Sec. \ref{sec:conclusion} by highlighting additional related questions that may warrant future 
investigation. In the Appendix, we collect a number of useful properties that generic matrices enjoy as well as complete
proofs of two of our main theorems [Theorems~\ref{qubitb} and~\ref{special_cases}], which we deemed too lengthy for 
inclusion in the main text.

\section{Background and preliminaries}
\label{sec:back}

\subsection{Dissipative quasi-local stabilization: Prior results}

Consider a multipartite quantum system composed of $N$ distinguishable, finite-dimensional subsystems (or ``qudits''). 
The corresponding Hilbert space, $\mathcal{H} \simeq \mathds{C}^d$, has a tensor product structure given by 
\begin{equation}
\label{eq:sys}
\H = \bigotimes_{a=1}^N\H_a, \quad a = 1,\dots,N, \quad \text{dim}(\H_a) = d_a, \quad \text{dim}(\H) =d. 
\end{equation}
Let $\mathcal{B}(\H)$ denote the space of bounded linear operators acting on $\H$. For $X \in \mathcal{B}(\H)$, 
its adjoint operator is given by $X^{\dagger}$. Physical states are represented by density operators which form a set 
of trace-one, positive semi-definite operators denoted by $\mathcal{D}(\H) \subset \mathcal{B}(\H)$. When using a 
matrix representation of operators, $\mathds{M}^{d_1 \times d_2}$ shall denote the space of complex matrices in 
$d_1 \times d_2$ dimensions. 

If $\rho_0$ represents the initial state of the system, say, at time $t=0$, then for any $t>0$ the evolved state is given by 
a completely-positive trace-preserving (CPTP) linear map (or quantum channel) $\mathcal{T}_t$ acting on $\mathcal{B}(\H)$, 
namely \cite{nielsen},
\begin{equation*}
\rho(t) \equiv \mathcal{T}_t(\rho_0) = \sum_k M_k\rho_0M_k^\dagger, 
 \quad \sum_k M_k^\dagger M_k = I, \quad M_k \in \mathcal{B}(\H),
\end{equation*}
 where $I$ is the identity operator in $\mathcal{B}(\H)$. The operator-sum representation of $\mathcal{T}_t$ given above, although not unique, exists if and only if the map is CP, whereas the second equality enforces the TP property.
The evolution of a large class of quantum systems of physical interest can be described using \emph{Quantum Dynamical Semigroups} (QDS). A QDS is a  continuous, one-parameter family of CPTP maps $\{\mathcal{T}_t\}_{t\geq 0}$, with $\mathcal{T}_0 = \mathcal{I}$, the identity map on $\mathcal{B}(\H)$, and obeying the Markov property, $\mathcal{T}_t \circ \mathcal{T}_s = \mathcal{T}_{t+s}$ for $t,s \geq 0$. These conditions ensure the existence of a linear map $\L$ (or Liouvillian) that generates the semigroup via $\mathcal{T}_t = e^{t \L}$ and in terms of which the resulting dynamics may be represented in the canonical Gorini-Kossakowskii-Sudarshan-Lindblad form ($\hbar \equiv 1$) \cite{alicki,nielsen}:   
\begin{equation} 
\label{LME}
\dot{\rho}(t) \equiv \L[\rho(t)] = -i[H,\rho(t)]+\sum_k{ \Big(L_k\rho(t)L_k^\dagger - \frac{1}{2}\{L_k^\dagger L_k,\rho(t)\}\Big)}, 
\quad t \geq 0.
\end{equation}
Here, $H=H^\dagger$ is the system's effective Hamiltonian and $\{L_k\}$ are the Lindblad (noise) operators,
that are associated to the non-unitary component of the dynamics. In this paper we consider only {\em time-invariant} dynamics, so that $H,\{L_k\}$ do not depend on time. A given choice of $H$ and $\{L_k\}$ unambiguously identifies the QDS generator associated with it, which we denote by $\L(H,\{L_k\})$. Vice-versa, the choice of $H$ and $\{L_k\}$ that specify the action of a given $\L$ on $\mathcal{B}(\H)$ is not unique: any transformation of the form 
\begin{equation}
\label{gen_transform}
\overline{L}_k \equiv L_k+c_k I, \quad \overline{H} \equiv H-\frac{i}{2}\sum_k{(c_k^*L_k-c_kL_k^\dagger)}, \quad c_k \in \mathds{C},
\end{equation}
can be verified to leave the generator unchanged, that is, $\L(H,\{L_k\}) = \L(\overline{H},\{\overline{L}_k\})$. 

The lattice geometry or coupling topology in the multipartite system under consideration typically impose constraints on the 
structure of $\L$. Building on our previous work \cite{DQLS:p,FFQLS,QLS,FTS,discrete}, we introduce locality constraints by specifying subsets of subsystems that are allowed to be acted upon non-trivially by the dynamics:
\begin{define} 
For the multipartite system \eqref{eq:sys}, a \emph{neighborhood} $\N_j\subsetneq \{1,\dots,N\}$ is a set of indexes that labels a group of subsystems. A \emph{neighborhood structure} is a collection  of neighborhoods $\N = \{\N_j\}_{j=1}^M$  for some 
finite $M$.  
\end{define}  

\noindent 
The \emph{neighborhood complement} of $\N_j$, denoted by $\overline{\N}_j$, consists of subsystems that are not included 
in $\N_j$, such that $\N_j \cup \overline{\N}_j = \{1,\dots,N\}$. Given a quantum state $\rho\in \mathcal{D}(\H)$, its reduced density matrix (RDM) on the neighborhood $\N_j$ is given by $\rho_{\N_j} = \tr_{\overline{\N}_j}(\rho)$, where $\tr_{\overline{\N}_j}$ is the partial trace over the subsystems in $\overline{\N}_j$. In this work we consider only neighborhood structures that are \emph{complete}, that is, $ \bigcup_{j=1}^M \N_j = \{1,\dots,N\}$.  For fixed $\N$, 
a \emph{neighborhood operator} is an operator that acts non-trivially on only one neighborhood. For example, 
$X = X_{\N_j}\otimes I_{\overline{\N}_j}$, with $X_{\N_j}\in \mathcal{B}(\H_{\N_j})$ and $I_{\overline{\N}_j} = 
\otimes_{a \notin \N_j}I_a$, is a neighborhood operator on $\N_j$. 

\begin{define}
Given a neighborhood structure $\N$, the QDS generator $\L$ is \emph{Quasi-Local} (QL) if it may be expressed 
as a sum of neighborhood generators, that is, 
$\L = \sum_{j=1}^M{\L_{\N_j} \otimes \mathcal{I}_{\overline{\N}_j}},$
where $\mathcal{I}_{\overline{\N}_j}$ is the identity map acting on $\mathcal{B}(\H_{\overline{\N}_j})$.
\end{define}

\noindent 
Quasi-locality of the generator ${\cal L}$ implies \cite{FFQLS} that there exists a representation $\L \equiv \L(H,\{L_k\})$, 
such that the associated $H$ and $\{L_k\}$ both satisfy QL constraints, that is,  
$H = \sum_{j=1}^M H_{\N_j} \otimes I_{\overline{\N}_j}$ and $L_k = L_{k,\N_j} \otimes I_{\overline{\N}_j}$, for all $k$. 

As our main aim is to explore stabilizability properties by QL-constrained dynamics, we next recall the relevant stability 
notions:

\begin{define}
A state $\rho_d$ is \emph{Globally Asymptotically Stable} (GAS) if it is an invariant state of the generator ~\eqref{LME}, that is, $\L[\rho_d]=0,$ and for any initial state the evolution approaches $\rho_d$ in the long time limit, that is, 
$\lim_{t \to \infty} e^{\L t}[\rho_0] = \rho_d,$ for all  $\rho_0 \in \mathcal{D}(\H).$
\end{define}

\noindent 
As shown in \cite{QLS} (Corollary 1), if a generator $\L(H,\{L_k\})$ makes a target {\em pure state} $\rho_d=|\psi\rangle\langle\psi|$ GAS, then by a suitable transformation of $H$ and $\{L_k\}$ as specified in Eq.~\eqref{gen_transform},
the {\em same} generator can be represented in a \emph{standard form} $\L(\overline{H},\{\overline{L}_k\})$, 
such that $\overline{H}\pure = h\pure$ and $\overline{L}_k\pure = 0$, for all $k$. Notice that such a transformation preserves the QL structure of the Hamiltonian and the Lindblad operators. 
In view of this, the following definition may be given for pure-state stabilization: 
\begin{define} 
\label{DQLS:define}
A pure state $ \rho_d =|\psi\rangle\langle\psi| \in \mathcal{D}(\H)$ is \emph{Dissipatively Quasi-Locally Stabilizable} (DQLS) relative to a neighborhood structure $\N$ if there exists a QL generator in standard form $\L(\overline{H},\{\overline{L}_k\})$, with $\overline{H} \equiv 0$, for which $\rho_d$ is GAS. 
\end{define}

\noindent 
The above definition implies that if $\rho_d$ is DQLS, then $\rho_d\in \ker(\L(\overline{L}_k))$ for each $k$, which means 
that $\rho_d$ must be invariant under the dynamics relative to each neighborhood. 

Finally, in order to characterize the set of DQLS states, we shall need the following linear-algebraic tools \cite{FFQLS}:

\begin{define} 
\label{schmidt_span}
Consider two inner product spaces $V_A, \,V_B$. Given a vector $\vec{v} \in V_A \otimes V_B$ with Schmidt decomposition 
$\vec{v} = \sum_i{\mu_i \,\vec{a}_i \otimes \vec{b}_i}$, the \emph{Schmidt span} of $\vec{v}$ relative to $V_A$ is the subspace
$\Sigma_A(\vec{v}) \equiv \Span\{\vec{a}_i \in V_A : \vec{v} = \sum_i{\mu_i\, \vec{a}_i \otimes \vec{b}_i}, \, \, \vec{b}_i \in V_B\}.$
The \emph{extended Schmidt span} of $\vec{v}$, denoted by $\overline{\Sigma}_{A}(\vec{v})$, is constructed from 
$\Sigma_A(\vec{v})$ as $\overline{\Sigma}_{A}(\vec{v}) \equiv \Sigma_{A}(\vec{v}) \otimes V_B$.
\end{define}

\noindent 
For notational simplicity, from now on we represent a pure state $\rho = \dpure$ by the corresponding {\em normalized} state vector $\pure \in \H$. With a slight abuse of terminology, $\Sigma_{\N_j}(\pure)$ will be denoted as the Schmidt span of $\pure$ relative to the neighborhood $\N_j$ instead of the space $\H_{\N_j}$. By construction, notice that 
$$\Sigma_{\N_j}(\pure) = \supp (\rho_{\N_j}), \quad 
\overline{\Sigma}_{\N_j}(\pure) = \supp (\rho_{\N_j}\otimes I_{\overline{\N}_j}).$$

The first characterization of DQLS states has been provided in \cite{DQLS:p}, and is reported below in 
its equivalent form in terms of the Schmidt span:
 
\begin{thm} [\cite{DQLS:p,FFQLS}] 
\label{DQLS:th}
A multipartite pure state $\pure \in \H$ is DQLS relative to the neighborhood structure $\N = \{\N_j\}_{j=1}^M$ if and only if
\begin{equation}
\label{DQLS:eq}
\bigcap_{j=1}^M \overline{\Sigma}_{\N_j}(\pure) \equiv \H_0(\pure) = \Span\{\pure\}.
\end{equation}
\end{thm}

\noindent 
The proof of this theorem in \cite{DQLS:p} makes it clear that the target-dependent subspace $\H_0(\pure)$ 
defined by \eqref{DQLS:eq} is the {\em smallest} subspace containing $\pure$, which can be made GAS using QL 
dissipation alone (that is, a QL generator that has $\overline{H}\equiv 0$ in standard form).
We shall refer to $\H_0(\pure)$ as the \emph{DQLS subspace associated to $\pure$}, relative to $\N$. 
Clearly, $\pure$ is DQLS if and only if its corresponding DQLS subspace is one-dimensional. 

An equivalent characterization of DQLS states, which also provides them with a transparent physical interpretation,  
is also proved as a corollary to Theorem~\ref{DQLS:th} in \cite{DQLS:p}:

\begin{cor}[\cite{DQLS:p}] 
\label{FFQL}
A multipartite pure state $\pure \in \H$ is DQLS relative to $\N$ if and only if it is the unique ground state of a 
FF QL parent Hamiltonian respecting the same neighborhood structure.
\end{cor}

\noindent
Recall that a QL Hamiltonian $H=\sum_j H_j$ is FF if any ground state of $H$ is also a ground state of {\em each} $H_j$. 
A {\em canonical} FF QL parent Hamiltonian $H_{\pure}$ may be constructed from the state 
$\pure$ itself by letting $H_{\pure} = \sum_j{H_j},$ $H_j \equiv I - \Pi_{\Sigma_{\N_j} (\pure)} \otimes I_{\overline{\N}_j},$
where $\Pi_{\Sigma_{\N_j} (\pure)}$ is the (orthogonal) projector onto the Schmidt span $\Sigma_{\N_j} (\pure)$.

\medskip

\begin{remark}
As mentioned in the Introduction, we shall focus henceforth on continuous-time Markovian dynamics, 
described by Eq.~\eqref{LME}.  The relevant notions of quasi-locality and stabilization are also applicable, 
however, to multipartite systems that evolve according to discrete-time Markovian dynamics.  
In the general non-homogeneous (time-varying) case, the state evolution is then given by a series of 
CPTP maps $\{\mathcal{T}_t\},$ such that
$\rho_{t+1} = \mathcal{T}_t(\rho_{t}), $ $\rho_t \in \mathcal{D}(\H),$ $t= 0,1,2,\ldots, $
and each such map is constrained to be a neighborhood map, namely, for each $t$ we may write $\mathcal{T}_t = 
\mathcal{T}_{t,\, \N_j} \otimes {\cal I}_{\overline{\N}_j}$ for some $j$.
It has been shown in \cite{discrete} that the mathematical characterization of pure states that are stabilizable by 
QL discrete-time dynamics yields the {\em same} necessary and sufficient condition as that for the DQLS states in 
Eq.~\eqref{DQLS:eq}. Thus, it follows that our analysis is applicable to QL discrete-time stabilizing dynamics as well. 
\end{remark}

\subsection{Reformulation of the DQLS condition for general pure states}
\label{restate}

Our main aim is to study the DQLS property of generic pure states. 
Given $\pure \in \H$ and its DQLS subspace $\H_0(\pure)$ for a specified neighborhood structure $\N$, 
the key mathematical idea underlying our approach is to characterize $\H_0(\pure)$
in terms of sets of neighborhood operators acting on $\pure$. This characterization can be employed to reformulate the 
DQLS condition in Eq.~\eqref{DQLS:eq} for an arbitrary number $N$ of subsystems. For $N=3$, the restated DQLS condition can be further reformulated as a set of linear equations whose coefficients are determined by the entries of $\pure$ in the standard basis. The solution space to this set of equations is shown to be in correspondence with $\H_0(\pure)$ and hence directly determines the DQLS property of $\pure$. In Sec. \ref{sec:generictrip}, such a reformulation will be instrumental in leveraging the properties of generic pure states toward understanding their DQLS nature. 

\subsubsection{$N$-partite setting.} 
The desired characterization of the DQLS subspace of a pure state $\pure$ in terms of sets of neighborhood operators 
stems from the following:

\begin{thm}
 \label{DQLS:subspace}
Let $\pure \in \H$ be a multipartite pure state and $\H_0(\pure)$ the corresponding DQLS subspace relative to the neighborhood structure $\N = \{\N_j\}_{j=1}^M$. Another pure state $|\psi'\rangle$ belongs to  $\H_0(\pure)$ if and only if there exists  operators $X_{\overline{\N}_j} \in \mathcal{B}(\H_{\overline{\N}_j})$ such that:
\begin{equation} 
\label{psi'}
|\psi'\rangle = (I_{\N_j} \otimes X_{\overline{\N}_j})\pure,  \quad j=1,\dots,M.
\end{equation}
\end{thm}

\begin{proof}
To prove the forward implication, let $\pure = \sum_l{\mu_l |\eta_l\rangle_{\N_j}\otimes |\gamma_l\rangle_{\overline{\N}_j}}$ and $|\psi'\rangle = \sum_k{\sigma_k |\chi_k\rangle_{\N_j}\otimes |\zeta_k\rangle_{\overline{\N}_j}}$ be the Schmidt decomposition of the two pure states relative to the $\N_j|\overline{\N}_j$ bipartition. Since $|\psi'\rangle \in \H_0(\pure)$, then $|\psi'\rangle \in \overline{\Sigma}_{\N_j}(\pure), \forall j$. Hence, it follows that 
$\Span\{|\chi_k\rangle_{\N_j}\} = \Sigma_{\N_j}(|\psi'\rangle) \subseteq \Sigma_{\N_j}(\pure) = \Span\{|\eta_l\rangle_{\N_j}\},$ 
for all $j.$ Thus, we may express $|\chi_k\rangle_{\N_j} = \sum_l {c_{k,l} |\eta_l\rangle_{\N_j}}$, with $c_{k,l} \in \mathds{C}$, for all $k$. 
Choose the operator $X_{\overline{\N}_j}$ such that $X_{\overline{\N}_j} |\gamma_l \rangle_{\overline{\N}_j} = \sum_k c_{k,l} (\sigma_k/\mu_l) \, |\zeta_k \rangle_{\overline{\N}_j} $. This guarantees that $(I_{\N_j} \otimes X_{\overline{\N}_j})\pure = |\psi'\rangle $, for all $j$, as claimed.
Conversely, to prove that any $|\psi'\rangle$ that is expressed by Eq.~\eqref{psi'} belongs to $\H_0(\pure)$, note that since $I_{\N_j} \otimes X_{\overline{\N}_j}$ is a linear transformation acting non-trivially only on $\H_ {\overline{\N}_j}$, $\Sigma_{\N_j}((I_{\N_j} \otimes X_{\overline{\N}_j})\pure) \subseteq \Sigma_{\N_j}(\pure)$. 
Hence, for any $|\psi'\rangle$ satisfying Eq.~\eqref{psi'},  
\begin{equation*}
\Span \{|\psi'\rangle\} \subseteq \bigcap_j \overline{\Sigma}_{\N_j}(|\psi'\rangle) \subseteq \bigcap_j \overline{\Sigma}_{\N_j}(\pure) = \H_0(\pure), 
\end{equation*}
from which the stated result follows.
\end{proof}

\noindent 
As a corollary,  the anticipated restatement of the DQLS condition in Eq. \eqref{DQLS:eq} follows:

\begin{cor}
\label{DQLS:restate}
A multipartite pure state $\pure \in \H$ is DQLS relative to the neighborhood structure $\N =\{\N_j \}_{j=1}^M $ if and only if the only 
operators $\{X_{\overline{\N}_j}\} \subseteq \mathcal{B}(\H_{\overline{\N}_j})$ that satisfy Eq.~\eqref{psi'} 
act as a scalar multiple of the identity on their Schmidt spans $\Sigma_{\overline{\N}_j}$, for all $j$.
\end{cor}

\begin{proof}
Assume that $\pure$ is DQLS relative to $\N$. Thus, $\H_0 (\pure)= \Span\{\pure\}$ and any $|\psi'\rangle \in \H_0(\pure)$ is $\pure$ itself, up to a constant scalar factor. Following Eq.~\eqref{psi'}, this is possible only if $X_{\overline{\N}_j}$ act as the identity (up to a constant scalar factor) on the corresponding $\Sigma_{\overline{\N}_j}(\pure)$, for all $j$.
To show the reverse implication, assume that the only choice of $\{X_{\overline{\N}_j}\}$ satisfying Eq. \eqref{psi'} is indeed the identity 
(up to a constant scalar factor) on $\Sigma_{\overline{\N}_j}(\pure)$, for all $j$. This means that any $|\psi'\rangle \in \H_0$ is proportional to $\pure$ and hence the latter is DQLS.
\end{proof}

The above results are useful to establish an important feature of the DQLS property of pure states, namely, its invariance under the action of SLOCC transformations, through which a given quantum state may be converted into another with a non-zero probability of success \cite{SLOCC}. Recall that two pure states are said to be in the same {\em SLOCC class} when they are related by an invertible local transformation of the form $\otimes_{a=1}^N{X_a}$, where $X_a \in \mathcal{B}(\H_a)$ for each $a$ \cite{tripartite}\footnote{Such a transformation need not be norm-preserving. However, the DQLS behavior of a pure state can be verified to be independent of its norm. This is further discussed in Remark~\ref{norm}. }. The following result then holds:
\begin{prop}
\label{SLOCC:DQLS}
The dimension of the DQLS subspace $\H_0(\pure)$ of a multipartite pure state $\pure \in \H$, 
relative to any fixed neighborhood structure, is preserved under arbitrary SLOCC transformations. 
\end{prop}

\begin{proof}
Let the two pure states $\pure, |\phi\rangle \in \H$ belong to the same SLOCC class. This means they are related as
$|\phi\rangle = (\otimes_{a=1}^N{X_a})\pure,$ where each $X_a \in \mathcal{B}(\H_a)$ is an invertible operator.
Thanks to the linearity of $X_a$, the Schmidt spans of these two states are related as $\Sigma_{\N_j}(|\phi\rangle) = 
(\otimes_{a\in \N_j}{X_a})\Sigma_{\N_j}(\pure)$ for any neighborhood $\N_j$ and, correspondingly, 
$\overline{\Sigma}_{\N_j}(|\phi\rangle) = (\otimes_{a=1}^N {X_a})\overline{\Sigma}_{\N_j}(\pure)$ for the extended Schmidt spans. 
Therefore,
\begin{equation}
\label{H_0-SLOCC}
\bigcap_{j=1}^M\overline{\Sigma}_{\N_j}(|\phi\rangle) = \Big(\bigotimes_{a=1}^N {X_a}\Big)
\bigcap_{j=1}^M\overline{\Sigma}_{\N_j}(\pure),
\end{equation}
whereby it follows that the corresponding DQLS subspaces relative to $\N = \{\N_j\}_{j=1}^M$ are related as $\H_0(|\phi\rangle) = (\otimes_{a=1}^N {X_a}) \, \H_0(\pure)$. 
Thanks to the invertible nature of all these transformations, $\dim (\H_0(|\phi\rangle)) = \dim (\H_0(\pure))$.
\end{proof}
\noindent 
As a corollary, if $\pure$ is DQLS, then any $|\phi\rangle$ belonging to the same SLOCC class is also DQLS relative to the chosen 
$\N$. In the non-DQLS case, the proof establishes the relation between the DQLS subspaces of two pure states that are in the same SLOCC class through Eq.~\eqref{H_0-SLOCC}. Notice that the converse of the above proposition is clearly not true, as the 
example of a W-state and GHZ-state on 3 qubits illustrates: these two states have  two-dimensional DQLS subspaces relative to 
any non-trivial neighborhood structure \cite{DQLS:p}, yet they do not belong to the same SLOCC class \cite{tripartite}. 

\begin{figure}[t]
\centering
\includegraphics[width=0.45\textwidth]{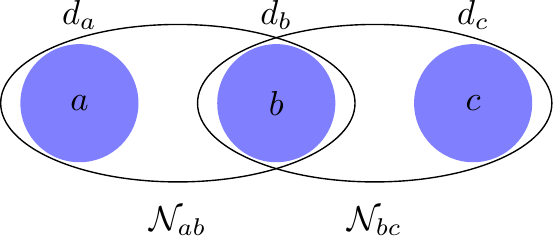}
\caption{Neighborhood structure $\N_{\rm tri} = \{\N_{ab},\N_{bc}\}$ relevant to the tripartite setting. 
We assume that the subsystem dimensions obey $d_c\geq d_a$, with $\bar{d}\equiv d_c-d_a$.}
\label{fig:tripartite}
\end{figure}

\subsubsection{Tripartite setting.}
\label{reform}
We now specialize to tripartite quantum systems, with the corresponding Hilbert space $\H \equiv \H_a \otimes \H_b \otimes \H_c$ 
and dimensions $d_a \times d_b \times d_c$, where we assume $d_c \geq d_a$ without loss of generality. The relevant neighborhood structure is denoted by $\N_{\rm tri} \equiv \{\N_{ab},\N_{bc}\}$ (see also Fig.~\ref{fig:tripartite}).
Let $\pure \in \H$. Then, by Theorem~\ref{DQLS:subspace}, the operators  $X_a \in \mathcal{B}(\H_a), X_c\in \mathcal{B}(\H_c)$ that satisfy the equation
\begin{equation} 
\label{tri:DQLS}
(X_a \otimes I_b \otimes I_c) \pure = (I_a \otimes I_b \otimes X_c) \pure,
\end{equation}
determine the DQLS subspace of $\pure$ relative to $\N_{\rm tri}$. As long as $\pure$ is fixed, Eq. (\ref{tri:DQLS}) is linear in 
$\{X_a,X_c\}$. The operators $X_a,X_c$ can thus be treated as unknowns and solved for, in order to characterize the DQLS 
property of $\pure$. For this purpose, we are going to rewrite Eq.~\eqref{tri:DQLS} as a set of linear equations in a selected 
basis, and establish an explicit connection between its solution space and $\H_0(\pure)$.

Let $\{|h\rangle_a\}$, $\{|i\rangle_b\}$ and $\{|j\rangle_c\}$ denote standard (orthonormal) bases in $\H_a$, $\H_b$ and $\H_c$, 
respectively. Decomposing $\pure$ with respect to the central subsystem $b$ yields 
\begin{equation} 
\label{phi_i}
\pure = \sum_{i=0}^{d_b-1}{|\phi_i\rangle_{ac}\otimes|i\rangle_b}, 
\end{equation}
from which we may rewrite Eq.~\eqref{tri:DQLS} in terms of $\{|\phi_i\rangle_{ac}\}$ as 
\begin{equation} 
\label{phi2}
(X_a \otimes  I_c) |\phi_i\rangle_{ac} = (I_a  \otimes X_c) |\phi_i\rangle_{ac}, \quad i=0,\dots ,d_b-1.
\end{equation}
In what follows, we will denote by ${}_{b}\langle \cdot |\psi\rangle$ the partial inner product with respect to the subsystem $b$, 
that is defined as ${}_{b}\langle i|\psi\rangle \equiv \sum_{h,j}{[{}_{ac}\langle hj |\otimes {}_{b}\langle i|\psi\rangle]\, |hj\rangle_{ac} = |\phi_i\rangle_{ac}}$, where $|hj\rangle_{ac} \equiv |h\rangle_a \otimes |j\rangle_{c}$, and similarly for the other subsystems.

Let $\vec{v}$ denote the vector form of $|v\rangle$ in the standard basis for the corresponding Hilbert space, 
and define a partial transpose operation $\mathcal{P}_T: \mathds{C}^{d_1} \otimes \mathds{C}^{d_2} \to 
\mathds{M}^{d_1\times d_2}$ as
\begin{equation*}
\mathcal{P}_T \Big(\sum_{ij}\lambda_{ij} \vec{u} \otimes\vec{v}\Big) \equiv \sum_{ij}\lambda_{ij}  \vec{u}\,
\vec{v}^{\,T}, \quad  \vec{u} \in \mathds{C}^{d_1},  \vec{v} \in \mathds{C}^{d_2}, \lambda_{ij} \in \mathds{C},
\end{equation*}
with $\vec{v}^{\, T}$ denoting its transpose in the standard basis. It is immediate to verify that $\mathcal{P}_T (X\vec{u}\otimes\vec{v}) = X(\vec{u} \,\vec{v}^{\,T})$ and $\mathcal{P}_T(\vec{u}\,\otimes \,Y\vec{v}) = (\vec{u}\, \vec{v}^{\,T})Y^T$, where $X \in  \mathds{M}^{d_1\times d_1}$,  $Y \in  \mathds{M}^{d_2\times d_2}$ and the matrix transpose is taken in the standard basis.
Let now $\vec{\phi}_i$ denote the vector representation of $|\phi_i\rangle_{ac}$ in the standard basis for $\H_a \otimes \H_c$, 
and $ \vec{\phi}_i = \sum_{k=1}^{\bar{d}_i}{\sigma_{i,k}\vec{u}_{i,k} \otimes \vec{v}_{i,k}}$ its Schmidt-decomposed form, with the corresponding Schmidt rank being given by $\bar{d}_i$. Applying the partial transpose operation on $\vec{\phi}_i$ with respect to 
the $a|b$ bipartition yields
\begin{equation} 
\label{A_i}
A_i \equiv \mathcal{P}_T (\vec{\phi}_i) = \sum_{k=1}^{\bar{d}_i}{\sigma_{i,k}\,\vec{u}_{i,k}\, \vec{v}_{i,k}^{\,T}}, 
\quad A_i \in  \mathds{M}^{d_a \times d_c}.
\end{equation}
It is easy to verify that Eq.~\eqref{A_i} is the singular value decomposition of $A_i$. For this reason, 
$\text{range}(A_i) = \Span\{\vec{u}_{i,k}\}$ and $\text{range}(A_i^T) = \Span\{\vec{v}_{i,k}\}$.  Thus, 
range$(A_i)= \Sigma_a(|\phi_i\rangle_{ac})$ and, similarly, range$(A_i^T)=\Sigma_c(|\phi_i\rangle_{ac})$, where the Schmidt spans are represented in standard basis for the respective Hilbert spaces. 

Next, recall Eq.~\eqref{phi2}. Thanks to the linear nature of $\mathcal{P}_T$, 
$\mathcal{P}_T((X_a\otimes I_c)\vec{\phi}_i) = X_aA_i$ and  $\mathcal{P}_T((I_a\otimes X_c)\vec{\phi}_i) = A_i X_c^T$, 
we may rewrite Eq.~\eqref{tri:DQLS} in terms of the matrices $\{A_i\}$ as
\begin{equation} 
X_a \, A_i = A_i \,X_c^T, \quad i = 0,\dots, d_b-1,
\label{conditions}
\end{equation}
with $\{X_a,X_c\}$, in matrix form, treated as unknowns. This makes it clear that 
the nature of the solution space of this linear system depends upon the DQLS property of the input state $\pure$. 
However, the original DQLS condition given in Eq.~\eqref{tri:DQLS} is based on the action of the operators 
$X_a \in \mathcal{B}(\H_a), X_c \in \mathcal{B}(\H_c)$ {\em restricted to the Schimdt spans} $\Sigma_a(\pure), \Sigma_c(\pure)$, respectively, by virtue of Theorem~\ref{DQLS:subspace}. At the same time, solutions to Eq. (\ref{conditions})
describe the action of $X_a,X_c$, in matrix form, {\em restricted to the combined} range of $\{A_i\}$ and $\{A_i^T\}$,
that is, $\Span\{\text{range}(A_i)\}_{i=0}^{d_b-1}$ and $\Span\{\text{range}(A_i^T)\}_{i=0}^{d_b-1}$, respectively.
Hence, in order to formally show that Eq.~\eqref{tri:DQLS} is equivalent to Eq. (\ref{conditions}), 
we need to further show that the corresponding subspaces indeed coincide. 

\begin{lemma} 
\label{row-column}
With $A_i$ defined as in Eq. (\ref{A_i}), the following equalities hold:
$$ \Sigma_a(\pure) = 
\Span\{{\rm{range}}(A_i)\}_{i=0}^{d_b-1}, \quad \Sigma_c(\pure) =
\Span\{{\rm{range}}(A_i^T)\}_{i=0}^{d_b-1}.$$
\end{lemma}

\begin{proof}
For notational simplicity, we denote subspaces like $\Sigma_a(|\phi_i\rangle_{ac})$ by their vector representation 
in the standard basis.
Following Eq.~\eqref{A_i}, we already know that range$(A_i)$ = $\Sigma_a(|\phi_i\rangle_{ac})$. Therefore, $\Span\{\text{range}(A_i)\}_{i=0}^{d_b-1} = \Span\{\Sigma_a(|\phi_i\rangle_{ac})\}_{i=0}^{d_b-1}$. Similarly, 
$\Span\{\text{range}(A_i^T)\}_{i=0}^{d_b-1} = \Span\{\Sigma_c(|\phi_i\rangle_{ac})\}_{i=0}^{d_b-1}$. 
Following Eq.~\eqref{phi_i}, $|\phi_i\rangle_{ac} = {}_{b}\langle i\pure$. If  $\pure = \sum_l{\mu_l \,|\gamma_l\rangle_{a} \otimes |\eta_l\rangle_{bc}}$ in the Schmidt-decomposed form for the $a|bc$ bipartition, then $|\phi_i\rangle_{ac} = \sum_l{\mu_l\,|\gamma_l\rangle_a}\otimes ({}_{b}\langle i|\eta_l\rangle_{bc})$ for all $i$. Thus, it follows that 
$$ \Sigma_a(|\phi_i\rangle_{ac}) \subseteq \Span \{|\gamma_l\rangle_a : {}_{b}\langle i|\eta_l\rangle_{bc} \neq \vec{0}_c\} \subseteq \Sigma_a(\pure), \quad i= 0,\dots , d_b-1,$$
where $\vec{0}_c$ denotes the null vector in $\H_c$. 
Therefore, $\Span\{\Sigma_a(|\phi_i\rangle_{ac})\}_{i=0}^{d_b-1} \subseteq \Sigma_a(\pure)$. To show that the two subspaces indeed coincide, assume that a vector $|v\rangle_a \in \H_a$ exists such that $|v\rangle_a \in \Sigma_a(\pure)$ but 
$|v\rangle_a \perp \Span\{\Sigma_a(|\phi_i\rangle_{ac})\}$. The latter implies that ${}_{a}\langle v|\phi_i\rangle_{ac} = \vec{0}_c$, for all $i$, which leads to ${}_{a} \langle v\pure = \vec{0}_{bc}$, thus contradicting our assumption. 
Hence, $\Span\{\Sigma_a(|\phi_i\rangle_{ac})\} = \Sigma_a(\pure)$.
This further shows that $\Span\{\text{range}(A_i)\}_{i=0}^{d_b-1} = \Sigma_a(\pure)$.  In a similar way, it can be shown that   
$\Span\{\text{range}(A_i^T)\}_{i=0}^{d_b-1} = \Sigma_c(\pure)$.
\end{proof}
We can now combine the above lemma with Corollary~\ref{DQLS:restate} to obtain the desired reformulation of the DQLS 
condition for general tripartite pure states: 
\begin{prop}
\label{DQLS:reform}
Consider a tripartite pure state $\pure \in \H_a \otimes \H_b \otimes \H_c$, and let the set of matrices $\{A_i\}$ 
be constructed from $\pure$ according 
to Eqs. \eqref{phi_i} and~\eqref{A_i}. Then $\pure$ is DQLS relative to $\N_{\rm tri}= \{\N_{ab},\N_{bc}\} $ if and only if 
the only solutions $\{X_a, X_c\}$ to the linear set of conditions
$X_a A_i = A_i X_c^T,$ $i = 0, \dots , d_b -1,$
act as a scalar multiple of the identity matrix on $\Span\{{\rm{range}}(A_i)\}_{i=0}^{d_b-1}$ and 
$\Span\{{\rm{range}}(A_i^T)\}_{i=0}^{d_b-1}$, respectively.
\end{prop}

\noindent
In this way, we have recast the problem of determining the DQLS property of a tripartite pure state $\pure$ 
into the one of characterizing the set of matrices $X_a$ and $X_c$ that solves the linear system in 
Eq.~\eqref{conditions}. The latter is guaranteed to have at least one (non-zero) solution, which corresponds to 
$X_a$ and  $X_c$ being proportional to the identity when restricted to the appropriate subspaces. The existence of any 
other non-trivial solution implies that $\pure$ is {\em not} DQLS. Thus, the structure of the solution space of  Eq.~\eqref{conditions} determines whether or not $\pure$ is DQLS. We next show that the {\em dimension} of the solution space is in fact 
the only relevant parameter in the case of generic tripartite states.

\section{Stabilizability properties of generic tripartite pure states}
\label{sec:generictrip}

In this section, we study the implications of the reformulated DQLS condition given in Proposition~\ref{DQLS:reform} 
in the context of generic tripartite pure states. 

\subsection{Random and generic tripartite pure states: Basic features} 
\label{generic}

A \emph{random pure state} in a finite-dimensional Hilbert space $\H \simeq \mathds{C}^d$ is akin to a 
random variable that is sampled from the set of pure states according to the (unique) unitarily invariant Fubini-Study measure \cite{wootters_random,bengtsson}. 
Random states may be generated as the orbit of a reference pure state $|\phi\rangle \in \H$ 
under the action of a random, Haar-distributed unitary operator $U \in  \mathcal{U}(d)$, where $\mathcal{U}(d)$ is the unitary 
group of degree $d$. The entries of a random (un-normalized) pure-state vector relative to any basis are i.i.d. complex Gaussian 
random variables \cite{rand_review}. As the Fubini-Study measure does not make any reference to the underlying factorization 
of $\H$, note that the sampling described above gives rise, in turn, to random pure states without a particular tensor structure.
A \emph{generic pure state} $\pure \in \H$ is a {\em typical} element in the set of random pure states, in the sense that it possesses all the properties that hold for a measure one set of states with respect to the Fubini-Study measure. Thus, the collection of such states is itself a measure one subset of pure states with respect to the same measure. In view of this, to establish whether random pure states are DQLS (or not) with probability one, it suffices to prove that a generic $\pure$ is DQLS (or not). In particular, with probability one a generic $N$-partite pure state has maximal Schmidt rank with respect to {\em any} bipartition of the subsystems, 
and the corresponding Schmidt coefficients are distinct \cite{lloyd} -- both of which properties will be heavily used in the following. 

A random tripartite pure state can be given a decomposition with respect to the central subsystem $b$ according to Eq.~\eqref{phi_i}. Notice that the entries of the vector representation of each  $|\phi_i\rangle_{ac}$ are a subset of the entries of the 
random pure-state vector with respect to the standard basis in $\H_a \otimes \H_b \otimes \H_c$. For this reason, each 
$|\phi_i\rangle_{ac}$, in this basis, is a random complex vector in $d_a\times d_c$ dimensions. Thus, $\{|\phi_i\rangle_{ac}\}$ 
correspond to a set of random pure (un-normalized) vectors in $\H_a \otimes \H_c$ that are linearly independent on each other. 
Recall that each $A_i$ is a reshaped version of the $|\phi_i\rangle_{ac}$ in its Schmidt basis, given by Eq.~\eqref{A_i}. 
Accordingly, $A_i$ is a random matrix of appropriate dimension. Further to that, if $\pure $ is a generic 
state in $\H$, then, by construction, $\{A_i\}$ is a set of linearly independent, generic matrices in $\mathds{M}^{d_a \times d_c}$ 
(see also Appendix \ref{sec:generic}).

\medskip 

\begin{remark} 
\label{norm}
We have so far dealt with {\em normalized} states $\pure$. For this reason, each state $|\phi_i\rangle_{ac}$ that is derived from $\pure$ using  Eq.~\eqref{phi_i} has norm at most equal to one. For the corresponding matrices $\{A_i\}$, this translates into the sum of the squares of their singular values obeying $\sum_k\sigma_{i,k}^2 \leq 1$, for all $i$.
While a generic matrix need {\em not} satisfy this condition, this problem can be circumvented by noticing that the 
characterization of DQLS states given in Eq.~\eqref{DQLS:eq} is independent upon normalization, as 
$\Span\{\pure\} = \Span\{\alpha\pure, \:\forall \alpha \in \mathds{C}\}$. Indeed, it can be verified that $\H_0(\alpha\pure) = \H_0(\pure)$.
The reformulated DQLS condition for tripartite states given by Proposition~\ref{DQLS:reform}  reflects this property: $\pure \mapsto \alpha\pure$ implies that $\{A_i\} \mapsto \{\alpha A_i\}$ and the set $\{X_a,X_c\}$ that satisfies $X_aA_i = A_iX_c^T$ also satisfies $X_a(\alpha A_i) = (\alpha A_i)X_c^T$. Thus, each $A_i$ that is derived from a normalized generic $\pure$ can be viewed as a representative of the family of generic matrices $\alpha A_i$ and the restriction on the sum of its singular values is irrelevant.
\end{remark}

\subsection{General results for arbitrary tripartite settings}

\subsubsection{No-go results and generic nature of the DQLS property.}
Our first result is a no-go theorem for the DQLS property of generic tripartite pure states, valid under certain combination of the subsystem dimensions. 

\begin{thm} 
\label{No-Go}
Consider a generic tripartite pure state $\pure \in 
\H_a \otimes \H_b \otimes \H_c$, where $d_a \leq d_c$ without loss of generality. $\pure$ is not DQLS relative to 
$\N_{\rm tri} = \{\N_{ab},\N_{bc}\}$ if $d_ad_b \leq d_c$. The dimension of the DQLS subspace $\H_0(\pure)$ is $d_a^2$. 
\end{thm}

\begin{proof}
The Schmidt rank of $\pure$ for the $ab|c$ bipartition is 
$\min\{\dim (\H_{ab}), \dim(\H_c)\} = \min\{d_ad_b,d_c\}$, 
generically. If $d_ad_b \leq d_c$, then $\dim(\Sigma_{ab}(\pure)) = d_ad_b$ and  $\Sigma_{ab}(\pure)=\H_a \otimes \H_b$. As a result, the extended Schmidt span $\overline{\Sigma}_{ab}(\pure) = \Sigma_{ab}(\pure) \otimes \H_c = \H$, the entire Hilbert space itself. The DQLS subspace is 
now $\H_0(\pure) = \overline{\Sigma}_{ab}(\pure) \cap \overline{ \Sigma}_{bc}(\pure)=\overline{ \Sigma}_{bc}(\pure)$, the smallest of the two subspaces. On the other hand, the Schmidt rank of the $a|bc$ bipartition is $d_a$, generically ($d_a \leq d_c$ assumed). Therefore, 
$\overline{ \Sigma}_{bc}(\pure)$ is a $d_a^2$-dimensional subspace. Thus, $\dim (\H_0) = d_a^2>1$ which violates the necessary and sufficient condition for DQLS in Theorem~\ref{DQLS:th}.
\end{proof}

The above result shows that as the difference between the outer subsystems' dimensions $d_a$ and $d_c$ gets bigger, 
generic pure states tend to be non-DQLS. An immediate corollary is that DQLS pure states, for $d_ad_b \leq d_c$, form a 
measure zero set. While for tripartite systems that do not fall under the above no-go condition (that is, for which $d_a d_b >d_c$), 
the analysis is more involved, we next show that, remarkably, the DQLS behavior or its absence {\em is}, nonetheless, a generic feature. 
The result builds on a number of properties we collect in the following:
\begin{prop}
\label{DQLS:generic}
	Consider a generic tripartite pure state $\pure \in \H_a \otimes \H_b \otimes \H_c,$ $\N_{\rm tri} = \{\N_{ab},\N_{bc}\}$ and, without loss of generality, 
	$d_a \leq d_c.$  If the subsystem dimensions obey $d_a d_b > d_c$, then:
	\begin{enumerate}[i.]
		\item $\Sigma_{a}(\pure) = \H_a$ and $\Sigma_{c}(\pure) = \H_c$. 
		\item $\pure$ is DQLS relative to $\N_{\rm tri}$ if and only if $\{X_a,X_c\}$ that satisfy Eqs. \eqref{conditions}, 
		act as a scalar multiple of the identity on $\H_a$ and $\H_c$, respectively.
		\item The elements of the DQLS subspace $\H_0(\pure)$ relative to $\N_{\rm tri}$ are in 1-1 correspondence with pairs $\{X_a,X_c\}$  
		of solutions of Eqs. \eqref{conditions}, hence the resulting dimensions are equal.
	\end{enumerate}
\end{prop}

\noindent
{\em Proof.}
	We prove each statement separately.  
		\begin{enumerate}[{\em i.}]
		\item  The Schmidt rank of a generic $\pure$ for the $a|bc$ bipartition is $d_a$, since $d_a \leq d_c$. Thus, 
		$\Sigma_{a}(\pure) = \H_a$. Due to the extra assumption $d_ad_b > d_c$, the Schmidt rank of the $ab|c$ bipartition 
		is equal to $d_c$, generically, and therefore, $\Sigma_{c}(\pure) = \H_c$.
		
		\item Combining Lemma~\ref{row-column} with the result just proved, the combined range of $\{A_i\}$, namely, 
		$\Span\{{\rm range}(A_i)\}_{i=1}^{d_b-1}$, coincide with $\H_a$. Similarly, the combined range of $\{A_i^T\}$ coincides 
		with $\H_c$. Hence, it follows from Proposition~\ref{DQLS:reform} that $\pure$ is DQLS relative to $\N_{\rm tri}$ if and only if 
		$X_a$ and $X_c$ act as the identity (up to a constant scalar factor) on $\H_a$ and $\H_c$, respectively.
		
		\item In Corollary~\ref{DQLS:restate}, we proved that $|\psi'\rangle \in \H_0(\pure)$ if and only if it satisfies 
		Eq. \eqref{psi'} for some 
		choice of $X_{\overline{\N}_j} \in \mathcal{B}(\H_{\overline{\N}_j})$, for all $j$. This observation, applied to the tripartite case 
		with $\N_{\rm tri}=\{\N_{ab},\N_{bc}\}$, implies that $|\psi'\rangle \in \H_0(\pure)$ if and only if there exist some 
		$X_a \in  \mathcal{B}(\H_a)$ and $X_c \in  \mathcal{B}(\H_c)$ such that 
		\begin{equation}
		|\psi'\rangle = (X_a\otimes I_{b} \otimes I_c) \pure =  (I_{a} \otimes I_b \otimes X_c) \pure.
		\label{psi'tri}
		\end{equation}
		What remains to be shown is that any fixed choice of $\{X_a,X_c\}$ that satisfies this equation can be identified with a 
		{\em unique} $|\psi'\rangle \in \H_0(\pure)$, and vice-versa. The forward implication follows straightforwardly from linearity.
		To prove the reverse implication, assume that there exist two different choices,  $\{X_a,X_c\}$  and  $\{Y_a,Y_c\}$, such 
		that they both satisfy 
		Eq.~\eqref{psi'tri} 
		for the same $|\psi'\rangle$. Restructuring this equation leads to the form given in Eq.~\eqref{conditions}. 
		After eliminating $|\psi'\rangle$, we get
		\begin{equation}
		(X_a-Y_a)A_i = \textbf{0}, \quad
		A_i(X_c-Y_c)^T = \textbf{0}, \quad  i = 0,\dots, d_b-1.\label{eq:psi'} 
		\end{equation}
		The first equation results in $X_a = Y_a$. This is because a generic matrix 
		$A_i \in \mathds{M}^{d_a \times d_c}$, with $d_a \leq d_c$, has no non-trivial left kernel. Now, form another matrix 
		\begin{equation*}
		\mathcal{A} = \left[\begin{array}{c}
		A_0 \\ \vdots \\ A_{d_b-1} \end{array} \right],
		\end{equation*}
		which is also generic, by construction. The second equation in Eq. (\ref{eq:psi'}) 
		then implies $\mathcal{A}(X_c-Y_c)^T = \textbf{0}$. Since 
		$ d_ad_b>d_c$,  $\mathcal{A}\in \mathds{M}^{d_ad_b \times d_c}$  has no non-trivial right kernel, hence  $X_c = Y_c$. 
		Therefore, for a given $|\psi'\rangle \in \H_0(\pure)$, there exists a unique choice of $\{X_a,X_c\}$ that satisfies 
		Eq.~\eqref{psi'tri}, whereby it also follows that the dimensions of these two spaces are equal. \hfill$\Box$
	\end{enumerate}

Notice that the assumption 
$d_a d_b > d_c$ 
is key to show that the Schmidt spans of a generic $\pure$ on subsystems $a$ and $c$ coincide with the entire space. 
Had that not been the case, the search for any non-trivial solutions to Eqs. \eqref{conditions} would have to be restricted to some subspace of the corresponding Hilbert spaces, making the analysis more involved. Although statement (ii) in Proposition~\ref{DQLS:generic} is essentially a restatement of Proposition~\ref{DQLS:reform} for {\em generic} tripartite pure states, combining statements (ii) and (iii) allows us to determine the {\em dimension} of the DQLS subspace of the target state, in addition to its DQLS behavior. Altogether, we are now in a position to prove the anticipated result on the generic nature of the DQLS property itself in those tripartite settings that do not meet the no-go condition: 

\begin{thm}
\label{generic:proof}
Consider a tripartite quantum system on $\H= \H_a \otimes \H_b \otimes \H_c$, with neighborhood structure $\N_{\rm tri} = \{\N_{ab},\N_{bc}\}$ and $d_a \leq d_c$, without loss of generality. When the subsystem dimensions obey the condition $d_a d_b > d_c$, DQLS pure states relative to $\N_{\rm tri}$  form either a  measure zero or a measure one set.  
\end{thm}

\begin{proof}

Consider a generic $\pure \in \H$. Since $d_a d_b > d_c$ and $d_a \leq d_c$, by Proposition~\ref{DQLS:generic} we know that 
$\Sigma_{a}(\pure) = \H_{a}$, $\Sigma_{c}(\pure) = \H_{c}$, and $\pure$ is DQLS relative to $\N_{\rm tri}$ if and only if Eq.~\eqref{psi'tri} is satisfied for $X_a = \lambda I_a$ and $X_c =  \lambda I_c$, with $\lambda \in \mathds{C}$. Define $\mathcal{T}  \equiv \Span \{X_a \otimes I_b \otimes I_c, I_a \otimes I_b \otimes X_c \,\vert \, X_a \in \mathcal{B}(\H_a), X_c \in \mathcal{B}(\H_c)\}$, which is a $(d_a^2+d_c^2-1)$-dimensional subspace of $\mathcal{B}(\H)$. In terms of $\mathcal{T}$, 
$\pure$ is DQLS relative to $\N_{\rm tri}$ if and only if 
$Y \pure =0$ for $Y \in \mathcal{T}$ implies that $Y = \textbf{0}$. Consider the representation of $\pure$ and $X_a,X_c$ in 
the standard basis for the respective spaces. If we vectorize $Y$ by stacking up its columns into a single vector denoted ${\rm vec}(Y)$, and use the property ${\rm vec}( AXB) =  (B^T\otimes A){\rm vec}(X)$ \cite{bhatia}, the condition $Y\pure=0$ is equivalent to
\begin{equation}
\label{M(psi)}
M(\vec{\psi}\,)^T {\rm vec} (Y) = {0},\quad \textup{where} \quad M(\vec{\psi}\,) \equiv \vec{\psi} \otimes I_{a} \otimes I_b \otimes I_c.
\end{equation}
Thus, the existence of any non-trivial ${\rm vec} (Y)$ in the right kernel of $M^T(\vec{\psi}\,)$ is equivalent to 
the non-DQLS nature of the generic $\pure$ which satisfies our assumption. To check whether this is the case, let us re-parametrize
${\cal T}$ by introducing an isometric embedding $V : \mathds{C}^{d_a^2+d_c^2-1} \rightarrow 
\text{vec}(\mathcal{T})$, such that \(V(\vec{e}_i)=\text{vec}(t_i),\)
where $\{t_i\}$ denotes a fixed basis in $\mathcal{T}$ and $\{\vec{e}_i\}$ is the standard basis in 
$\mathds{C}^{d_a^2+d_c^2-1}$. 		
Using the fact that $V$ is full (column-) rank, we have that $Y\neq \textbf{0}$ can satisfy Eq.~\eqref{M(psi)}, for a pure state $\pure,$ 
if and only if $A\equiv M(\vec{\psi})^T V$ has a non-trivial right kernel. 
This is equivalent to saying that $A^\dagger A$ is non-invertible. Therefore, a generic $\pure$ fails to be DQLS if and only if   
\begin{equation} 
\label{det}
D(\psi) \equiv {\rm det}[(M(\vec{\psi}\,)^TV)^\dagger M^T(\vec{\psi}\,)V]= 0.
\end{equation}
The above may be viewed as a {\em polynomial equation}, whose coefficients are determined by $V$ and 
whose unknowns are given by the entries of $\dpure$ in the standard basis. The roots of a polynomial in $\C^n$, 
for any $n$, are a measure zero set whenever the latter is non-zero, or a measure one set for the trivial (zero) polynomial.
In the first case, when $D(\psi)\geq 0$, since Eq. \eqref{det} implies that $Y= \mathbf{0}$ is the only solution 
to Eq. \eqref{M(psi)} for all but a zero measure set of instances, the DQLS property is measure one. If $D(\psi)\equiv 0$, 
all $V$-dependent coefficients vanish identically, hence non-trivial solutions $Y$ always exists and DQLS is, 
correspondingly, a measure zero property.
\end{proof}

Since the proof of the above theorem hinges on the fact that, for $d_ad_b>d_c$, the Schmidt spans of a 
generic $\pure$ with respect to the $a|bc$ ($ab|c$) bi-partitions coincide with the full $\H_a$  ($\H_c$), 
we can extend some of our conclusions to {\em any} general tripartite state $|\phi\rangle$, whose Schmidt spans on the relevant bi-partitions are maximal.  This is equivalent to saying that the corresponding one-body RDMs, $\rho_{a} \equiv \rho_{\overline{\N}_{bc}}$ and $\rho_{c} \equiv \rho_{\overline{\N}_{ab}}$, are full-rank. 
Observe that when DQLS is a measure zero property, the polynomial 
$D(\phi)$ defined in Eq.~\eqref{det} from a state $|\phi\rangle$ (generic or otherwise), whose RDMs $\rho_{a}$ and $\rho_{c}$ are full-rank, is also identically zero. We can thus give the following corollary:
\begin{cor}
\label{full_red}
Consider a tripartite quantum system on $\H=\H_a \otimes \H_b \otimes \H_c$ and neighborhood structure $\N_{\rm tri} =\{\N_{ab},\N_{bc}\} $, such that the subsystem dimensions obey $d_ad_b>d_c$ and $d_a \leq d_c$. If a generic pure state $\pure \in \H$ is DQLS relative to $\N_{\rm tri}$ with  probability zero, then an arbitrary pure state $|\phi\rangle \in \H$, with full-rank RMDs 
$\rho_{a} \in \mathcal{D}(\H_a)$ and $\rho_c \in \mathcal{D}(\H_c)$, is also not DQLS relative to $\N_{\rm tri}$.
\end{cor}

\subsubsection{SLOCC canonical form of generic pure states.} 
\label{sec:SLOCC}

While Theorem~\ref{generic:proof} guarantees that the DQLS property is either measure zero or measure one 
in the set of pure states, further analysis is required to characterize the behavior in cases where 
$d_a d_b > d_c$. Toward this, we leverage the SLOCC-invariance of the DQLS property to simplify our analysis,
 by transforming the given generic state into another state with useful properties. The following lemma will be needed:
\begin{lemma}
\label{max_entangle}
Any bipartite pure state $|\phi\rangle \in \H_1 \otimes \H_2$ with maximal Schmidt rank may be transformed into the maximally entangled (un-normalized) state, namely, $|\Omega\rangle = \sum_k |e_k\rangle_1 | f_k\rangle_2,$ 
through the action of local invertible operators. Here, $\{|e_l\rangle\}_{l=1}^{d_1}$ and $\{|f_m\rangle\}_{m=1}^{d_2}$ 
represent the standard bases in $\H_1$ and $\H_2$, respectively.

\end{lemma} 
\begin{proof}
The maximal Schmidt rank of a bipartite state is equal to the dimension of the smallest component Hilbert space in the bipartition. Assume $d_1 \leq d_2$, without loss of generality, and let the Schmidt decomposition of $|\phi\rangle$ be given by $|\phi\rangle = \sum_{k=1}^{d_1}\mu_k|x_k\rangle_1|y_k\rangle_2$. Construct the linear operator $M_1\in\mathcal{B(H}_1)$ such that $M_1=\sum_{k=1}^{d_1}(1/\mu_k) |e_k\rangle\langle x_k|$. This is invertible since $\{|x_k\rangle\}_{k=1}^{d_1}$ forms a basis of $\H_1$, thanks to the maximal Schmidt rank of $|\phi\rangle$. Now, choose $M_2 \in\mathcal{B(H}_2)$ to be the unitary that transforms the orthogonal basis $\{|y_m\rangle\}_{m=1}^{d_1}$ (by extending the Schmidt vectors appropriately) into $\{|f_k\rangle\}_{k=1}^{d_2}$ in $\H_2$.  
\end{proof}

We make use of this fact to convert a generic state $\pure$ into another pure state 
$(M_a\otimes I_b \otimes M_c)\pure$, in such a way that 
\begin{equation} 
\label{phi0}
(M_a \otimes M_c )|\phi_0\rangle_{ac} = |\Omega\rangle_{ac},
\end{equation}
where $|\phi_0\rangle_{ac}$ belongs to the decomposition of $\pure$ given by Eq.~\eqref{phi_i}. Lemma~\ref{max_entangle} applies  because $|\phi_0\rangle_{ac}$  has maximal Schmidt rank, owing to its generic nature. Let $\{\tilde{A}_i\}$ denote the set of matrices associated to $(M_a\otimes I_b \otimes M_c)\pure$, derived using Eq.~\eqref{phi_i} and Eq. \eqref{A_i}. With the modified 
$|\phi\rangle_{ac}$ being maximally entangled, these matrices have the form:
\begin{equation}
\label{SLOCC_matrices}
\tilde{A}_0 = \left[\begin{array}{c|c}
I_{00} & \textbf{0}_{0\bar{d}} \end{array} \right], \quad \tilde{A}_i = M_aA_iM_c^T, \quad i = 1, \dots, d_b-1,
\end{equation} 
where the block $I_{00}$ denotes the identity matrix in $\mathds{M}^{d_a \times d_a},$ corresponding to the Schmidt coefficients of the maximally entangled state, and $\textbf{0}_{0\bar{d}}$ denotes the zero matrix in $\mathds{M}^{d_a \times \bar{d}}$, with 
$\bar{d} \equiv d_c-d_a$. 
Since the dimension of $\H_0(\pure)$ and $\H_0(M_a \otimes I_b \otimes M_c\pure)$ are the same thanks to Proposition~\ref{SLOCC:DQLS}. Hence, in order to determine the DQLS nature of $\pure$ we can study the set of linear equations~\eqref{conditions} associated to $\{\tilde{A}_i\},$ instead of $\{A_i\}$. 

The form of $\tilde{A}_0$ in Eq.~\eqref{SLOCC_matrices} considerably simplifies the problem of solving the set of equations~\eqref{conditions} for $\{\tilde{A}_i\}$. For this reason, and for the sake of simplicity, we hereafter denote by 
{\em $\pure$ an SLOCC-transformed generic tripartite pure state}, such that the corresponding 
$|\phi_0\rangle_{ac} = |\Omega\rangle_{ac}$ in Eq.~\eqref{phi_i}, with $|\Omega\rangle_{ac}$ being the maximally 
entangled state. We also drop the tilde in denoting $\{A_i\}$, the set of matrices constructed from the transformed 
$\pure$ following Eqs.~\eqref{phi_i}-\eqref{A_i}.
In order to analyze the form of the solutions, it is convenient to introduce a block decomposition of the matrices involved in Eq.~\eqref{conditions}. Since $X_a$ is the smallest matrix, it is not further decomposed. 
The decomposition of the other matrices is summarized in Table~\ref{block}. 
\noindent 
Eq.~\eqref{conditions}, with $A_0 = \left[\begin{array}{c|c} I_{00} & \textbf{0}_{0\bar{d}} \end{array}\right],$ 
simplifies the form of the matrix $X_c$ as shown below:
\begin{equation}
\tilde{X}_{00} = X_a, \quad   \tilde{X}_{0\bar{d}} = \textbf{0}_{0\bar{d}}. 
\label{X_c1}
\end{equation}

\begin{table}[t] 
\centering
\begin{tabular}{|c|c|c|}
\hline 
{\sf Matrix} & {\sf Dimensions} & {\sf Block form}  \\ 
\hline\hline
\multirow{2}{*}{$A_i$, $i = 1,\dots,d_b-1$} & \multirow{2}{*}{$d_a \times d_c$} & \multirow{2}{*}{$\left[\begin{array}{c|c} A^i_{00} & A^i_{0\bar{d}} 
\end{array} \right]$} \vspace*{0mm}\\ 
[3ex] 
\multirow{2}{*}{$X_c^T$ } & \multirow{2}{*}{$d_c \times d_c$} & \multirow{3}{*}{$\left[\begin{array}{c|c} \tilde{X}_{00} & \tilde{X}_{0\bar{d}} \\ \hline \\ \tilde{X}_{\bar{d}0} &\ \tilde{X}_{\bar{d}\bar{d}} \end{array} \right]$} \\[8ex]
\hline
\end{tabular}
\caption{\label{block}The subscripts used to denote the matrix blocks are chosen based on the following logic: $00$,  $\bar{d}\bar{d}$, $0\bar{d}$ and  $\bar{d}0$ denote $d_a \times d_a$, $\bar{d} \times \bar{d}$, $d_a \times \bar{d}$ and $\bar{d} \times d_a$ dimensions for the blocks, respectively. Also, $\bar{d} \equiv d_c - d_a \geq 0$ following the assumption $d_c \geq d_a$. }
\end{table}

\subsection{Analytical characterization for a central qubit subsystem, $d_b=2$}
\label{sub:qubit}

We now provide analytical results that fully characterize the DQLS nature of generic tripartite pure states when the central subsystem 
is a qubit. In this case, the set of relevant equations \eqref{conditions} are associated to the matrices $A_0$ and $A_1$.  
Since we are considering the SLOCC-transformed generic pure state, $A_0$ is equal to identity padded 
with zeros for adjusting the dimensions.
As a consequence, Eqs.~\eqref{X_c1} hold for the different blocks of $X_c$. Now apply Eqs.
\eqref{conditions} to $A_1$ to obtain the following conditions: 
\begin{align} 
X_aA_{0\bar{d}} &= A_{0\bar{d}}\tilde{X}_{\bar{d}\bar{d}}, \label{cqubit1} \\
X_aA_{00} &= A_{00}X_a + A_{0\bar{d}}\tilde{X}_{\bar{d}0}, \label{cqubit2}
\end{align}
where, for notational simplicity, we have dropped the superscript 1 while denoting the blocks of $A_1$.
The case $d \times 2 \times d$ ($\bar{d}=0$) is easier to analyze as 
the relevant matrices $X_a,X_c,A_0$ and $A_1$ have the same dimensionality. As a result, 
in the decomposition of Table~\ref{block}, 
all the blocks except the $d_a \times d_a$ ones vanish, and Eqs.~\eqref{cqubit1}-\eqref{cqubit2} reduce to
\begin{equation}
\label{commute}
X_{a}A_{00} = A_{00}X_{a},
\end{equation}
with $X_c^T = X_a$ (notice that $A_{00}$ is the same as $A_1$, but we stick to the former notation for consistency). 
In the more general scenario $d_a\times 2\times d_c$, where $\bar{d} = d_c-d_a > 0$, we know from Theorem~\ref{No-Go} 
that generic pure states are non-DQLS when $d_a \leq \bar{d}$, in which case the corresponding DQLS subspace 
$\H_0$ is $d_a^2$-dimensional. As we are going to show, the only situation in which the DQLS property is found to be 
generic corresponds to the combination of dimensions $d_a \times 2\times (d_a+1)$. More precisely, our 
analytical characterization is summarized in the following theorem, whose lengthy proof we postpone, for clarity, to Appendix \ref{sub:proofs}:

\begin{thm}
\label{qubitb}
Let $\pure$ be a generic tripartite pure state in 
$ \H_a \otimes \H_b \otimes \H_c$, with dimensions $d_a \times d_b \times d_c$ and 
neighborhood structure $\N_{\rm tri} = \{\N_{ab},\N_{bc}\}$. When $d_b = 2$, 
\begin{enumerate}[i.]
\item $\pure$ is not DQLS relative to $\N_{\rm tri} $ when $d_a = d_c$. The dimension of $\H_0(\pure)$ is $d_a$.
\item $\pure$ is not DQLS relative to $\N_{\rm tri} $ when $d_c>d_a$, except for $d_c = d_a +1$. The dimension 
of $\H_0(\pure)$ is $\min\{d_a^2,(d_c-d_a)^2\}$.
\end{enumerate}
\end{thm}

\subsection{Stabilizability properties beyond $d_b=2$} 
\label{sub:beyond}

Consider next a general tripartite systems in $d_a \times d_b \times d_c$ dimensions, for $d_b \ne 2$. 
Numerical evidence shows that the dimension combinations which do not fall under the no-go condition of Theorem~\ref{No-Go} 
yield a positive result to the DQLS test -- with {\em some} exceptions. It appears as if this is an instance of an over-constrained system of 
equations, with Eq.~\eqref{conditions} having more constraints than in the $d_b =2$ cases, whereas the number of unknowns remain the same. This observation would suggest to pursue a constraint-counting analysis similar to the one used in proving Theorem~\ref{qubitb}. However, notwithstanding the fact that the number of constraints in Eq.~\eqref{conditions} increases with $d_b$, they remain linearly dependent on each other to a large extent. This is verified by the fact that the system is always guaranteed to admit at least one non-trivial solution, namely, $X_a$ and $X_c$ proportional to the identity matrix. In view of this intricate general structure, we focus on providing analytic results for some specific cases of interest, and use numerical tools to investigate the remaining ones.

\subsubsection{Analytical results.} We begin with the following observation: 
\begin{prop}
\label{higher-d_b}
If a generic tripartite pure state in $d_a \times d_b \times d_c$ dimensions is DQLS relative to $\N_{\rm tri} = \{\N_{ab},\N_{bc}\}$, 
the same property holds in any $d_a \times d'_b \times d_c$ dimensions, with $d'_b >d_b$.
\end{prop}

\begin{proof}
It can be verified that none of the above dimension combinations fall under the no-go Theorem~\ref{No-Go}. The set of 
equations $X_aA_i = A_iX_c^T$ for $i=0,\dots,(d'_b-1),$ where $\{A_i\}$ denotes a set of generic matrices in 
$\mathds{M}^{d_a \times d_c}$, determines the DQLS nature of generic pure states in $d_a \times d'_b \times d_c$ 
dimensions, following Proposition~\ref{DQLS:generic}.
Since, by assumption, generic pure states in $d_a \times d_b \times d_c$ dimensions are DQLS, $\{X_a,X_c\}$ that simultaneously solve any $d_b$ of these equations are proportional to the identity in the corresponding space. Hence, this holds for the entire set too and thus the generic pure states in $d_a \times d'_b \times d_c$ dimensions are also DQLS, for $d'_b > d_b$.
\end{proof}

\noindent 
Note that the converse of the above proposition is generally {\em false}. A counter example is given by considering 
$2 \times 3 \times 2$ vs.  $2 \times 2 \times 2$: as we will soon prove (Theorem~\ref{special_cases}), 
generic pure states are DQLS relative to $\N_{\rm tri}$ in the former case, yet we know from 
Theorem~\ref{qubitb} that they are not in the latter.

For the remainder of our analysis, we consider Eqs.~\eqref{conditions} in terms of the block-decomposition of the relevant matrices 
(Table~\ref{block}). As {\em we will be using the SLOCC-transformed version of the generic pure state}, the matrices 
$\{A_i\}$ have $A_0 = \left[\begin{array}{c|c} I_{00} & \textbf{0}_{0\bar{d}} \end{array} \right]$.
The consequence of this transformation is to simplify the block structure of $X_c$ by setting 
$\tilde{X}_{00} = X_a$ and $\tilde{X}_{0\bar{d}} = \textbf{0}$. By taking into account this simplification,
we can now rewrite Eqs. \eqref{conditions} with respect to the block decomposition of the relevant matrices as follows:
\begin{align}
X_a A_{0\bar{d}}^{i} &= A_{0\bar{d}}^{i}\tilde{X}_{\bar{d}\bar{d}},
 \label{many1}\\
X_a A_{00}^i &= A_{00}^{i}X_a + A_{0\bar{d}}^{i}\tilde{X}_{\bar{d}0}, \quad i=1,\dots,(d_b-1).  
\label{many2}
\end{align}
Notice that for $\bar{d} \equiv d_c -d_a>0$, Eq.~\eqref{many1} looks similar to the set of equations~\eqref{conditions} with $\{A_i\}$, $\{X_a,X_c\}$ replaced by $\{A_{0\bar{d}}^i\}$, $\{X_a,\tilde{X}_{\bar{d}\bar{d}}^T\}$, respectively, and the system dimensions 
given by $d_a \times (d_b-1) \times \bar{d}$. Thus, we see that the DQLS property of a generic pure state in $d_a \times (d_b-1) \times \bar{d}$ dimensions influences the solution space of Eq.~\eqref{many1}, according to Proposition~\ref{DQLS:generic}. The proposition given below builds further on this observation.

\begin{prop}
\label{partDQLS}
If a generic tripartite pure state in $d_a \times (d_b-1) \times \bar{d}$ dimensions is DQLS relative to 
$\N_{\rm tri} = \{\N_{ab},\N_{bc}\}$, the same holds for a generic tripartite pure state in $d_a \times d_b \times (d_a+\bar{d})$ dimensions, for $0<\bar{d} < d_a(d_b -1)$.
\end{prop}

\begin{proof}
If the generic pure states in $d_a \times (d_b-1) \times \bar{d}$ dimensions are DQLS, then the only choice of $X_a$ and $X_c$ 
that satisfy Eqs. \eqref{conditions} is proportional to the identity. In order to determine the DQLS property of generic pure states in $d_a \times d_b \times (d_a+\bar{d})$ dimensions, let us analyze the modified version of Eqs. \eqref{conditions}, given by Eqs. 
\eqref{many1}-\eqref{many2}. Since Eq.~\eqref{many1} is identical to Eq.~\eqref{conditions} with $\{A_i\}$, $\{X_a,X_c\}$ replaced by $\{A_{0\bar{d}}^i\}$, $\{X_a,\tilde{X}_{\bar{d}\bar{d}}^T\}$, respectively, we see that the only choice of $\{X_a,\tilde{X}_{\bar{d}\bar{d}}\}$ that satisfy Eq.~\eqref{many1} must be the identity, following our assumption. As a consequence, Eq.~\eqref{many2} is modified to $A_{0\bar{d}}^{i}\tilde{X}_{\bar{d}0} = \textbf{0}$, for all $i$. Since we are only considering the cases where $d_ad_b > d_c$ 
(or, $\bar{d} < d_a(d_b -1)$), which are not covered under the no-go Theorem~\ref{No-Go}, the only choice of $\tilde{X}_{\bar{d}0}$ that fall in the right kernel of each generic matrix in $\{A_{0\bar{d}}^i\}$ is the zero matrix (as established at the end of the proof of statement (iii) in Proposition~\ref{DQLS:generic}). Let us now collect together all the above observations, in order to determine the solution space to Eqs.~\eqref{many1}-\eqref{many2} in $d_a \times d_b \times (d_a+\bar{d})$ dimensions. We already showed that $X_a = I$. Recall that, in block form, $X_c^T = \left[\begin{array}{c|c} 
X_a & \textbf{0} \\ \hline  
\tilde{X}_{\bar{d}0} &\ \tilde{X}_{\bar{d}\bar{d}} 
\end{array} \right].$ 
If $\tilde{X}_{\bar{d}0}$ and $\tilde{X}_{\bar{d}\bar{d}}$ are the zero matrix and the identity matrix, respectively, it follows that $X_c = I.$ Therefore, invoking Proposition~\ref{DQLS:generic}, generic pure states in $d_a \times d_b \times (d_a+\bar{d})$ dimensions are also DQLS.
\end{proof}

We are now in a position to characterize a number of special cases, which we collect in the following theorem.  
Since the proof is lengthy, we again defer it to Appendix \ref{sub:proofs}. 
\begin{thm}
\label{special_cases}
Let $\pure$ be a generic tripartite pure state in $\H_a\otimes\H_b \otimes \H_c$, with neighborhood structure 
$\N_{\rm tri} = \{\N_{ab},\N_{bc}\}$ and central subsystem dimension $d_b >2$. Then, $\pure$ is DQLS relative 
to $\N_{\rm tri}$ for the following combinations of subsystem dimensions:  
\begin{enumerate}[i.]
\item $d_a \times d_b \times d_a$, 
\item $d_a \times d_b \times (d_a+1),$
\item $d_a \times d_b \times n d_a$, with $1<n<d_b$.
\end{enumerate}
\end{thm}

\subsubsection{Numerical results.}
In order to investigate the cases where an analytical proof is lacking, we resort to numerical techniques, by 
generating random pure states in various dimensions\footnote{Specifically, we employed T. Cubitt's {\sf randPsi.m} 
function in \textsc{Matlab}, available online at {http://www.dr-qubit.org/matlab/randPsi.m}.}
and assessing their DQLS property relative to $\N_{\rm tri} = \{\N_{ab},\N_{bc}\}$. In light of Theorem~\ref{generic:proof}, 
we know that the set of DQLS states has either measure zero or measure one, hence a randomly generated state with probability one will have the same property of almost all its peers.  Table~\ref{numerical1} summarizes our numerical 
results in $d_a \times d_b \times d_c$ dimensions when the central subsystem is a qutrit, $d_b = 3$.
\begin{table}[t]
\centering
\begin{tabular}{|c|c|c|c|c|c|c|c|c|c|c|c|c|c|c|c|c|c|} \hline
\backslashbox{$d_a$}{$\bar{d}$} & 0 & 1 & 2 & 3 & 4 & 5 & 6& 7& 8& 9& 10 &11 & 12 \\ \hline 
2 & $\checkmark$ & $\checkmark$ & $\checkmark$ & $\checkmark$ & ${\bm \times}$ &  ${\bm \times}$ &${\bm \times}$ & ${\bm \times}$ & ${\bm \times}$   & ${\bm \times}$ &${\bm \times}$ & ${\bm \times}$ & ${\bm \times}$  \\
\hline 
3 & $\checkmark$ & $\checkmark$ & $\checkmark$ & $\checkmark$ & $\checkmark$ &  $\cirtick$  & ${\bm \times}$   & ${\bm \times}$  & ${\bm \times}$  & ${\bm \times}$  & ${\bm \times}$ & ${\bm \times}$  & ${\bm \times}$  \\
\hline 
4 & $\checkmark$ & $\checkmark$ & $\cirtick$ &$\checkmark$ & $\checkmark$ & $\checkmark$ & $\cirtick$  & $\circross$   &${\bm \times}$  & ${\bm \times}$ & ${\bm \times}$ & ${\bm \times}$  & ${\bm \times}$ \\
\hline 
5 & $\checkmark$  & $\checkmark$  & $\cirtick$  & $\cirtick$  & $\checkmark$  &$\checkmark$ & $\checkmark$ & $\cirtick$ & $\cirtick$  &  $\circross$  & ${\bm \times}$ &${\bm \times}$ &  ${\bm \times}$ \\ \hline 
\end{tabular} 
\caption{\label{numerical1} 
DQLS behavior of random pure states for $d_b =3$. $\bar{d} \equiv d_c - d_a \geq 0$  without loss of generality. 
A check mark indicates that the tested random pure states exhibit DQLS behavior, whereas cross mark means the opposite. 
Circled symbols highlight the DQLS behavior of those dimension combinations that are not covered by our analytical results. }
\end{table}
It can be seen that, for fixed values of $d_a$, increasing $\bar{d}$ results in a gradual shift from DQLS to non-DQLS behavior, 
as intuitively expected. This is not the case, however, when $d_b = 2$, whereby $d_a \times 2 \times (d_a +1)$ is the only dimension combination for which random states are DQLS, generically. With the exception of $4 \times  3 \times 11$ and $5 \times  3 \times 14$ dimensions, the transition from DQLS to non-DQLS behavior occurs precisely when $\bar{d} = (d_b-1)d_a$, that is, $d_ad_b = d_c$, which corresponds to the threshold for the applicability of the no-go Theorem~\ref{No-Go}. The above two dimension combinations 
are the only cases in this table that are not DQLS and do not fall under the no-go category. Additional numerical investigation of 
such ``exceptional cases'' suggest that (i) they do occur for $\bar{d}$ close to $(d_b-1)d_a$ (the no-go threshold) for $d_a > d_b$; 
and (ii) the number of such occurrences is approximately given by $d_a/d_b$. These trends are 
captured in Table~\ref{numerical2}, where we have chosen $d_b = 4$.

\begin{table}[t]
\centering
\begin{tabular}{|c|c|c|c|c|c|c|c|c|c|c|c|c|c|c|} \hline
\backslashbox{$d_a$}{$\bar{d}$} & 11 & 12 & 13 & 14 & 15 & 16 & 17& 18& 19& 20 & 21 & 22 & 23 & 24\\[10pt] \hline 
4 & $\checkmark$ & ${\bm \times}$ &${\bm \times}$ &  ${\bm \times}$ & ${\bm \times}$ &${\bm \times}$ &  ${\bm \times}$ & ${\bm \times}$ &${\bm \times}$ &  ${\bm \times}$ & ${\bm \times}$ &${\bm \times}$ &  ${\bm \times}$ & ${\bm \times}$ \\
\hline 
5 &  $\checkmark$ &$\checkmark$ & $\checkmark$ & $\circross$ & ${\bm \times}$ &${\bm \times}$ &  ${\bm \times}$ & ${\bm \times}$ &${\bm \times}$ &  ${\bm \times}$ & ${\bm \times}$ &${\bm \times}$ &  ${\bm \times}$ & ${\bm \times}$\\
\hline
6 &  $\checkmark$ & $\checkmark$ & $\checkmark$ & $\checkmark$ & $\checkmark$ &   $\checkmark$ & $\circross$ & ${\bm \times}$ &${\bm \times}$ &  ${\bm \times}$ & ${\bm \times}$ &${\bm \times}$ &  ${\bm \times}$ & ${\bm \times}$ \\
\hline
7 & $\checkmark$ & $\checkmark$ &$\checkmark$ & $\checkmark$ & $\checkmark$ &  $\checkmark$ & $\checkmark$ &$\checkmark$ & $\checkmark$ &$\circross$ & ${\bm \times}$ &${\bm \times}$ &  ${\bm \times}$ & ${\bm \times}$\\
\hline
8 &  $\checkmark$&$\checkmark$ &$\checkmark$ & $\checkmark$ &$\checkmark$ &  $\checkmark$ &$\checkmark$ & $\checkmark$ & $\checkmark$ & $\checkmark$ & $\checkmark$ & $\circross$ & $\circross$ & ${\bm \times}$\\
\hline
\end{tabular} 
\caption{
DQLS behavior of random pure states for $d_b = 4$, with symbols having the same meaning as in Table \ref{numerical1}. 
In particular, circled symbols now correspond to the dimensions where the tested target states are not DQLS, 
but they do not fall under the no-go category.}
\label{numerical2}
\end{table} 

Based on our combined analytical and numerical results, we conclude our analysis of the generic DQLS property of random 
tripartite pure states with the following: 

\smallskip

\noindent
{\bf Conjecture 3.11.} {\em Let $\pure$ be a generic pure state in $d_a \times d_b \times d_c$ dimensions, with $d_b >2$. Then: 
\vspace*{-1mm}  
\begin{enumerate}[i.]
\item When $d_a \leq d_b$, $\pure$ is generically DQLS relative to $\N_{\rm tri}$ if and only if $d_c < d_ad_b$.
\item When $d_a > d_b$, $\pure$ is generically DQLS relative to $\N_{\rm tri}$ if $d_c < d_ad_b - \floor{d_a/d_b}$.
\end{enumerate}}

\section{Additional results and implications}
\label{sec:more}

\subsection{Stabilizability properties in the generic multipartite setting}

The analytical tools developed in Sec.~\ref{restate} and the results in Sec. \ref{sec:generictrip}
pertain to generic tripartite pure states relative to the bipartite neighborhood structure $\N_{\rm tri} = \{\N_{ab},\N_{bc}\}$. 
Addressing multipartite quantum systems, with arbitrary neighborhood structures, is very hard. Nonetheless, building on  our tripartite analysis, it is still possible to identify conditions that {\em preclude} generic DQLS behavior in multipartite settings. 
First, by extending Theorem~\ref{No-Go}, a no-go condition for the DQLS property of generic multipartite pure states 
may be established, based on the sizes of the relevant neighborhoods in $\N$:
\begin{thm}
\label{No-Go:ext}
Consider a generic multipartite pure state $\pure \in \H$ and a neighborhood structure $\N = \{\N_j\}_{j=1}^M$. $\pure$ is not DQLS relative to $\N$ if $\dim(\H_{\N_j}) \leq \dim(\H_{\overline{\N}_j})$, for all $j$. Additionally, the DQLS subspace of $\pure$ coincides with $\H$.  
\end{thm}

\begin{proof}
For a generic pure state $\pure$, the Schmidt rank with respect to any bipartition is maximal. Therefore, for any 
$\N_j|\overline{\N}_j$ bipartition,  $\dim(\Sigma_{\N_j}(\pure))=\dim(\H_{\N_j})$, since $\dim(\H_{\N_j}) \leq \dim(\H_{\overline{\N}_j})$ by 
assumption. Hence, $\Sigma_{\N_j}(\pure) = \H_{\N_j}$ and $\overline{\Sigma}_{\N_j}(\pure) = \H$. The DQLS subspace in this case is $\H_0 (\pure)= \bigcap_{j=1}^M \overline{\Sigma}_{\N_j}(\pure) = \H$, implying that $\pure$ is not DQLS according to  Theorem~\ref{DQLS:th}. 
\end{proof}

\noindent 
As a consequence, it follows that in many systems of practical interest, where the number of subsystems is large but only 
finite-range (typically, nearest or next-to-nearest neighbor) couplings are considered, generic pure states are not DQLS, as 
one may intuitively expect. In fact, the above result also implies that, if one begins with an $N$-partite quantum system 
where generic states are DQLS relative to the chosen neighborhood structure, the sizes of the neighborhoods have to be scaled appropriately as $N$ increases, in order for the DQLS property to be retained.  

\medskip
\begin{remark}
\label{No-Go:non-generic}
Similar to Corollary \ref{full_red}, it is worth noting that the no-go Theorem~\ref{No-Go:ext} is applicable to any 
pure state $|\phi\rangle$ which has {\em full-rank RDMs} in all the neighborhoods under consideration. 
Since a full-rank RDM on neighborhood $\N_j$ is equivalent to saying that $\overline{ \Sigma}_{\N_j}(|\phi\rangle) = \H$, 
such states always fail to be DQLS. For instance, one may verify that qubit graph states have maximally 
mixed RDMs on any two-qubit neighborhoods and hence they are non-DQLS under nearest-neighbor interactions, 
consistent with known results \cite{Kraus2008,DQLS:p} and the fact that they cannot be exact ground states of 
two-body parent Hamiltonians \cite{graph}.  Another interesting class of states that are non-DQLS for the same reason 
comprises so-called $k$-uniform states, namely, $N$-party pure states such that their all $k$-body RDMs for	
$k \leq \ceil{N/2}$ are maximally mixed \cite{k-uniform}.
\end{remark}

\medskip

To derive other useful tests for the DQLS property by leveraging our tripartite results, we need a way to relate multipartite and multi-neighborhood settings to the known tripartite scenario. The key observation is that DQLS properties 
are preserved when aggregating, or coarse-graining, these neighborhoods:
\begin{define} 
Let $\pure \in \H=\otimes_{a = 1}^N \H_a$, where each $\H_a$ is a finite-dimensional Hilbert space, with 
neighborhood structure $\N = \{\N_j\}_{j=1}^M$.  A neighborhood structure $\N'=\{\N'_k\}_{k=1}^{M'}$ is a \emph{coarse-graining} of $\N$ if for each $\N_j \in \N$ there exists some $\N'_k \in \N'$, such that $\N_j \subseteq \N'_k$. Equivalently, $\N$ can be called a \emph{refinement} of $\N'$. 
\end{define}

\noindent 
We next show that the DQLS property of a general pure state relative to a given $\N$ is preserved if we move to any 
coarse-graining of the same, as expected:
\begin{prop}
\label{refinement}
If the multipartite pure state $\pure \in \H$ is DQLS relative to the neighborhood structure $\N= \{\N_j\}_{j=1}^M$, 
then $\pure$ is DQLS relative to any coarse-graining of $\N$ as well.
\end{prop}

\begin{proof}
Let $\N'=\{\N'_k\}_{k=1}^{M'}$ denote the coarse-graining of $\N$. By Definition~\ref{DQLS:define}, $\pure$ is DQLS relative to $\N$ if and only if there exists a QL generator $\L(\overline{H},\{\overline{L}_k\})$ in standard form, with $\overline{H} = 0$, such that each term in $\L$ acts non-trivially only on a given neighborhood $\N_j \in \N$ and $\L$ makes $\dpure$ GAS. Since each $\N'_k \in \N'$ is such that $\N'_k \supseteq \N_j$ for some $\N_j \in \N$, $\L$ is QL with respect to the coarse-grained neighborhood structure $\N'$ as well. Therefore, the same QL generator renders $\pure$ DQLS relative to $\N'$.
\end{proof} 

\noindent 
This, in turn, implies that if $\pure$ is not DQLS relative to a certain neighborhood structure, then $\pure$ is {\em not} DQLS relative to any refinement as well. Given a multipartite generic pure state and an associated $\N$, one can then coarse-grain $\N$ into a two-neighborhood neighborhood structure $\N'$, similar to the one described in Fig.~\ref{fig:tripartite}, and analyze it in the light of the tripartite results given in the previous sections. If the generic pure state is not DQLS relative to $\N'$, then it is not DQLS relative to $\N$ also. 

We may organize all the results on the DQLS property of generic pure states into an algorithm 
that can decide the DQLS nature in many cases.  The algorithm takes as input a generic target pure state $\pure$ and a neighborhood structure $\N = \{\N_j\}_{j=1}^M$ on $\H$, and proceeds through the following steps:
\begin{enumerate} [(i)]
\item \emph{Check if there exists any neighborhood $\N_{\bar{k}} \in \N$ such that 
$\dim(\H_{\N_{\bar{k}}}) \leq \dim(\H_{\overline{\N}_{\bar{k}}})$. If so, form a modified neighborhood structure $\widetilde{\N}$ 
by ignoring all such neighborhoods}.
If there exists an $\N_{\bar{k}}$ satisfying the above condition, its corresponding extended Schmidt span 
$\overline{\Sigma}_{\N_{\bar{k}}}(\pure) = \H$, according to the reasoning given in the proof of Theorem~\ref{No-Go:ext}. As a result, $\overline{\Sigma}_{\N_{\bar{k}}}(\pure)$ contributes trivially to the DQLS subspace $\H_0 (\pure) = \bigcap_{j=1}^M \overline{\Sigma}_{\N_j}(\pure)$ relative to $\N$, and hence can be ignored in the analysis. In other words, for any neighborhood on which the Schmidt span of $\pure$ is full-rank, the invariance condition on the dissipative generator renders the corresponding neighborhood generator trivial, and therefore disposable (see Definition~\ref{DQLS:define}). This also shows that the DQLS subspaces of $\pure$ relative to $\N$ and $\widetilde{\N}$ coincide with each other. \emph{If all the neighborhoods in $\N$ can be eliminated in this way, then, as proved in Theorem~\ref{No-Go:ext}, $\pure$ is not DQLS relative to $\N$}. 

\item \emph{Check if there exists any pair of neighborhoods $\N_{\bar{k}},\N_{\bar{l}} \in \widetilde{\N}$, 
such that 
$\N_{\bar{k}} \cup \N_{\bar{l}} = \{1,\ldots, N\}$. Then if $\pure$ is DQLS relative to 
the neighborhood structure $\{\N_{\bar{k}},\N_{\bar{l}}\}$, $\pure$ is also DQLS relative to $\N$}.

\item As a result of ignoring the smaller neighborhoods in Step (i), \emph{if there exist subsystems 
that do not belong to any $\N_j \in \widetilde{\N}$, then $\pure$ is not DQLS relative to $\N$}. 
In such cases, the extended Schmidt spans of any of the remaining neighborhoods may be written as 
$\overline{\Sigma}_{\N_j} = (\Sigma_{\N_j} \bigotimes_{a \notin (\N_j \bigcup r)} \H_a ) \bigotimes_{a \in r} \H_a$,  
where $r$ collects all the subsystems that do not belong to any of the relevant neighborhoods. 
Then
\begin{equation*}
\H_0 (\pure)= \Big(\!\bigcap_{\N_j \in \tilde{\N}} \Sigma_{\N_j}\!\! \bigotimes_{a \notin (\N_j \bigcup r)} \!\! \H_a \Big) 
\bigotimes_{a \in r} \H_a,
\end{equation*}
and $\dim(\H_0(\pure))$ is at least equal to the product of the dimensions of the 
leftover subsystems. Since the latter exceeds one, 
we know that $\pure$ is not DQLS. 
\item \emph{Coarse-grain $\widetilde{\N}$ to form a two-neighborhood structure $\N'$, as 
in Fig.~\ref{fig:tripartite}. If $\pure$ is not DQLS relative to $\N'$, then $\pure$ is non-DQLS relative to $\N$ as well}. 
\end{enumerate}

One may have to resort to numerical techniques if this procedure fails to be conclusive, by explicitly determining whether 
$\H_0(\pure)$ is one-dimensional for the specified neighborhood structure on $\H$.

\subsection{Practical stabilization under QL constraints} 
\label{practical}

From a practical standpoint, a potential application of our results on the DQLS property of random pure states is 
to enable {\em arbitrarily accurate approximate stabilization} of target pure states that are otherwise non-DQLS. 
We illustrate the basic idea through a QIP-motivated example.   

Consider the 4-qubit GHZ state, $|\psi_{\rm GHZ}\rangle = (|0000\rangle +|1111\rangle)/\sqrt{2}$, with respect to 
three-body neighborhoods $\N = \{\N_{123}, \N_{234}\}$. The GHZ state belongs to the measure zero set of 
non-DQLS states in this setting \cite{DQLS:p}, although, as we know from  Theorem~\ref{special_cases}, random 
pure states of 4-qubits {\em are} generically DQLS relative to $\N$.  We now show that one can always find a DQLS 
pure state $|\psi_\epsilon\rangle$, such that $|\psi_{\rm GHZ}\rangle$ can be approximately stabilized by QL
dissipation with arbitrary accuracy $\epsilon >0$. To do so, pick a random unitary $U \in \mathcal{U}(16)$ with 
respect to the Haar measure and, for fixed $\epsilon$, define the integer $n\equiv \lceil 1/\epsilon\rceil$.  
Correspondingly, let $U_\epsilon \equiv U^{1/n}$ and $|\psi_\epsilon\rangle \equiv U_\epsilon |\psi_{\rm GHZ}\rangle$. 
Since the eigenphases of $U$ are uniformly distributed in $[0,2\pi)$, the eigenphases of $U_\epsilon$ are uniformly 
random within $[0,2\pi/ n)$, with the associated eigenvectors being unchanged \cite{mehta}. 
The set of pure states $|\psi_\epsilon\rangle$ that are generated from the unitaries following this modified distribution 
forms a {\em finite measure subset} of the pure states in $\H$, and it thus follows that is DQLS generically.  
We now claim that, with probability one, the preparation fidelity 
$F\equiv |\langle \psi_{\rm GHZ} | \psi_\epsilon\rangle|^2 =1 - c_U\epsilon + \mathcal{O}(\epsilon^2)$, for 
a positive constant $c_U$ depending on $U$. 
This follows from noticing that $U_\epsilon$ is approximately close to the identity operator acting on the $4$-qubit Hilbert 
space. If $H \equiv -i \text{Log}(U)$, where $\text{Log}(\cdot )$ restricts the eigenvalues of the Hermitian operator 
$H$ to $[0,2\pi)$, we may write $U_\epsilon = e^{-i \epsilon H}$.  By construction, 
$H$ is a bounded operator hence $|| H  ||_{\rm op}$ is finite, with $\epsilon || H  ||_{\rm op} \ll 1$ if $\epsilon$ is sufficiently 
small, as stipulated. By Taylor-expanding, we then have 
\begin{eqnarray*}
|| I - U_\epsilon ||_{\rm op} = || i \epsilon H + {(i\epsilon H)^2}/{2!} + \dots ||_{\rm op} \, \leq  \, 
\frac{\epsilon || H  ||_{\rm op}}{1-\epsilon || H  ||_{\rm op}}.
\end{eqnarray*}
Thus, $|| I - U_\epsilon  ||_{\rm op} \leq \epsilon || H ||_{\rm op}+\mathcal{O}(\epsilon^2)$. 
By using standard properties of operator norm, the desired result follows:
\begin{equation*}
F =  |1 - \langle \psi_{\rm GHZ} | (I - U_\epsilon) | \psi_{\rm GHZ} \rangle|^2 \, \geq\, 
\Big | 1 - \frac{\epsilon || H  ||_{\rm op}}{1-\epsilon || H  ||_{\rm op}} \Big |^2 \approx 1 -2 \epsilon ||H||_{\rm op} + 
{\mathcal O}(\epsilon^2).
\end{equation*}

A similar analysis may be carried out for any pure state that belongs to the zero measure set of non-DQLS states 
in a setting where pure states are generically DQLS.    
  
\subsection{Generic pure states as ground states of frustration-free Hamiltonians}

As we remarked [Corollary~\ref{FFQL}], DQLS states may be physically characterized as unique ground states of FF QL parent Hamiltonians. Thus, our results may be equivalently interpreted as identifying some QL notions under which {\em almost any} pure state may arise as the unique ground state of a FF Hamiltonian -- and, likewise, many other QL notions for which such a description is not possible. In the spirit of of practical stabilization as discussed above, it also follows that non-DQLS states in those settings for which random pure states are generically DQLS may be written as approximate unique ground states of some FF QL parent Hamiltonian. 

When generic states are not DQLS because of the no-go Theorem~\ref{No-Go:ext}, we can actually draw a stronger conclusion. 
As the proof of the above no-go shows, when neighborhoods are too small, in the sense that they all contain no more than half of the subsystems, the DQLS subspace coincides with the entire $\H$. We now show that the following property holds for all pure states:
\begin{prop} 
\label{ground_states}
Let $H = \sum_k{H_k} = \sum_k h_{\N_k} \otimes I_{\overline{\N}_k}$ be any FF Hamiltonian QL relative to $\N = \{\N_k\}$.
If $\pure$ is a ground state of $H$, its DQLS subspace $\H_0(\pure)$ is contained in the ground-state space of $H$.
\end{prop}

\begin{proof}
Without loss of generality, we can scale the ground-state energy of each $h_{\N_k}$ such that it coincides with zero, the same 
holding for $H$, thanks to its FF property. Since $H\pure = 0$ implies that $(h_{\N_k} \otimes I_{\overline{\N}_k})\pure = 0$, 
for all $k$,  we also have that $\text{tr}(h_{\N_k}\rho_{\N_k}) = 0$. Using the fact that both $h_{\N_k}$ and $\rho_{\N_k}$ are positive semi-definite, it also follows that 
\begin{equation}
\label{1}
\supp(\rho_{\N_k}) = \Sigma_{\N_k}(\pure) \subseteq \text{ker}(h_{\N_k}), \quad \forall \, k.
\end{equation}
Recall that the DQLS subspace of $\pure$ is given by 
$\H_0 (\pure)= \bigcap_k (\Sigma_{\N_k}(\pure)\otimes \H_{\overline{\N}_k})$. For any $|\phi\rangle \in \H_0(\pure)$, $|\phi\rangle \in  \Sigma_{\N_k}(\pure)\otimes \H_{\overline{\N}_k}$, for all $k$. 
It can then be verified that \(\Sigma_{\N_k}(|\phi\rangle) \subseteq \Sigma_{\N_k}(\pure),\) for all $k$.  Using Eq. \eqref{1}, 
this shows that $\Sigma_{\N_k}(|\phi\rangle) \subseteq \text{ker}(h_{\N_k})$, and hence $H|\phi\rangle = 0$, for any $|\phi\rangle \in \H_0(\pure)$. 
Accordingly, $\H_0(\pure) \subseteq \text{ker}(H)$.
\end{proof}	

\noindent 
From the above observations, we conclude that for the multipartite settings where Theorem~\ref{No-Go:ext} applies, 
{\em no} FF QL Hamiltonian other than a trivial (scalar) one can have a generic pure state in its ground-state space. 
This is to be expected in view of the fact that since generic states have full-rank neighborhood RDMs, each term 
$h_{\N_k}$ that has the corresponding RDM in its ground space can only be the identity, by Eq.~\eqref{1}. 

\medskip

\begin{remark}
We have already commented on the applicability of the extended no-go Theorem~\ref{No-Go:ext} to special classes of 
pure states with full-rank RDMs in Remark~\ref{No-Go:non-generic}. As a consequence of Proposition~\ref{ground_states}, 
these states are also {\em not} expressible as the exact ground states of any non-trivial FF QL Hamiltonians with appropriate neighborhood structure. Interestingly, the fact that qubit graph states are never exact ground states of non-trivial two-body FF Hamiltonians recovers a known result independently established in \cite{graph}. 
\end{remark}

\subsection{QL Stabilization beyond generic pure states}
\label{tdesign}

While our main emphasis has been on generic pure states, it is worth stressing that the key mathematical features upon which our conclusions rest may hold for different sets of target states. In the tripartite setting, Corollary~\ref{full_red} 
shows that when DQLS is a measure zero property relative to $\N_{\rm tri}$, arbitrary pure states with full-rank RDMs 
$\rho_{\overline{\N}_{bc}}, \rho_{\overline{\N}_{ab}}$ on subsystems $a$ and $c$ also behave the same way. Similarly, 
Remark~\ref{No-Go:non-generic} demonstrates that quantum states with full-rank neighborhood RDMs $\rho_{\N_j}$ fall under the reach of the extended no-go Theorem~\ref{No-Go:ext} and thus are non-DQLS.
 
In a similar spirit, since generating random pure states on $\H \simeq \mathds{C}^d = (\mathds{C}^2)^{\otimes N}$ 
is well known to entail resources that scale exponentially in $N$ \cite{bengtsson}, ``pseudo-random'' quantum states have been extensively analyzed in QIP \cite{Joseph2003,caves,Gross,Winton}, that share only some properties with the former yet suffice for relevant tasks. Loosely speaking, a quantum state $t$-design is a subset of states that approximates the Haar distribution in that a pure state sampled from the design is indistinguishable from a Haar-sampled random pure state, given access to at most $t$ copies of the state \cite{caves}. More convenient to our scope is an equivalent definition of a $t$-design in terms of polynomials on pure-state amplitudes \cite{ambainis}. Let $\pure  = \sum_{i=1}^d \alpha_i | i\rangle$ with respect to the standard basis for $\mathds{C}^d$, and let $P_{(t,t)}(\psi) \equiv P(\alpha_1,\dots , \alpha_d, \alpha^*_1, \dots, \alpha^*_d)$ be a complex polynomial in $2d$ variables that is {\em balanced} in the sense that it has degree at most $t$ in both $\alpha_1, \ldots,\alpha_d$ and $\alpha_1^*, \ldots, \alpha_d^*$. A probability distribution $(p_i, |\phi_i\rangle)$ is called an {\em exact quantum state $t$-design} if the expected value of any $t$-order balanced polynomial  $P_{(t,t)}(\psi)$ over the set coincides with the expectation over {\em all} (normalized) pure states in $\H$,   
\begin{equation*}
\int_{\psi} P_{(t,t)} (\psi)\,d\psi= \sum_i p_i \, P_{(t,t)} (\phi_i), \quad \sum_i p_i =1, 
\end{equation*}
where the integral is taken with respect to the Haar-induced (Fubini-Study) measure.

Our claim is that, for sufficiently large $t$ and suitable subsystem dimensions, 
pure target states sampled from a $t$-design share the DQLS properties of 
fully random pure states, to high probability. In essence, this follows from the fact that suitable {\em concentration bounds} hold for the statistical moments of $t$-designs \cite{low_paper}. Consider first a tripartite system on $\H  = \H_a \otimes \H_b \otimes \H_c$, with neighborhood structure $\N_{\rm tri} = \{\N_{ab},\N_{bc}\}$, and let  $|\phi_t\rangle \in  \H$ be sampled from a $t$-design. The first relevant bound constraints the entropy across any bipartition: specifically, by using Lemma 2.7 in 
\cite{low_paper} for an exact design ($\epsilon=0$), one may show that the probability for the von Neumann entropy of subsystem $a$ across the bi-partition $a|bc$ (say, $S(\rho_a))$, to deviate from its maximum value by more than $\delta >0$ is upper-bounded as 
\begin{equation} 
{\mathbb P}[ |S(\rho_a)-S_{\text{max}}(\rho_a)| \geq \delta ] \leq 4 \Big( \frac{d_a }{C d_b d_c} t \Big)^{t/8}, 
\label{tbound}
\end{equation}
where $\delta = \alpha + d_a/(d_b d_c \ln 2 ), \alpha >0$, and the constant $C \approx (4/9\pi^3) (2^\alpha -1)^2$. The above bound implies that, with high probability, $|\phi_t\rangle$ has full-rank RDM on subsystem $a$, with a similar result holding across the bi-partition $c|ab$. Similarly, Theorem 1.2 in \cite{low_paper} shows that, if $P_{(k,k)}(\cdot)$ is any balanced polynomial of degree at most $k$, its expectation is close to its Haar expectation value with high probability, provided that $t\geq 2 k$.

Suppose now that the subsystem dimensions satisfy $d_c < d_a d_b$ (note that $d_a < d_b d_c $ is also automatically satisfied 
in our setting), so that by Theorem~\ref{generic:proof} DQLS states form either a measure zero or a measure one set. 
In the former case, thanks to the bound in Eq. \eqref{tbound}, Corollary~\ref{full_red} is applicable to 
$|\phi_t\rangle$ and it thus follows that the state is not DQLS, 
with high probability.
On the other hand, when the DQLS property is generic relative to $\N_{\rm tri}$, the polynomial $D(\cdot)$ defined in 
Eq. (\ref{det}) is different from zero. By construction, $D(\phi_t)$ is a balanced polynomial of degree at most $(d_a^2+d^2_c-1)$ 
in terms of $|\phi_t\rangle$ in the standard basis, thus its expectation value over the $t$-design coincides 
with that of the Haar measure, as long as $t \geq (d_a^2+d^2_c-1)$. Since this expectation value is strictly positive when Haar random pure states are DQLS with probability one, $D(\phi_t)$ is non-zero, with high probability due to Theorem 1.2 in \cite{low_paper}. 
It then follows that $|\phi_t\rangle$ is DQLS relative to $\N_{\rm tri}$, 
with high probability, as anticipated. By a similar line of reasoning, it can also be verified that for suitable subsystem-dimensions and large enough $t$, states 
sampled from a $t$-design are non-DQLS with high probability, whenever the no-go Theorems~\ref{No-Go} and~\ref{No-Go:ext} 
are applicable in the tripartite and multipartite cases, respectively.

\section{From quasi-local stabilization to uniquely determined quantum states}
\label{sec:marginal}

\subsection{DQLS vs UDA pure states}

As mentioned in the Introduction, the task of stabilizing a globally specified quantum state 
in the presence of QL constraints is naturally related to aspects of the quantum marginal problem 
and the general theme of how ``parts'' relate to the ``whole'' in multipartite settings. Mathematically, a direct 
connection between a DQLS state and its set of neighborhood RDMs is established in 
Theorem~\ref{DQLS:th}: the intersection of the extended Schmidt spans or, equivalently, 
the supports of the RDMs (extended to a global state by appending identities of the proper dimensions) 
must return the support of the target alone.  Since the relevant set of RDMs in the DQLS problem is derived 
from the given target pure state $\pure$, the ``existence'' part of the quantum marginal problem is automatically satisfied. 
Notwithstanding, knowledge about the QL stabilizability properties of $\pure$ may shed light onto whether and how the state 
may be uniquely reconstructed solely based on QL information.  

Following established terminology \cite{Chen}, given a list of density operators $\{\rho_\ell \}$, we say that 
$\pure$ is \emph{Uniquely Determined among all Pure states} (UDP) or \emph{Uniquely Determined among All states} (UDA) if 
no other pure or, respectively, arbitrary quantum state has the same $\{\rho_\ell\}$ as its RDMs 
(or marginals); 
for example, $k$-body RDMs have the form $\rho^{i_1\ldots i_k} = \tr_{i_{k+1}\ldots i_N}(\pure \langle\psi|)$, $i_j \in \{1,\ldots,N\}$.
While the UDA property clearly implies UDP, it was proved recently that the converse does {\em not} hold true 
in general \cite{UDAvsUDP}. 
It is also known that a DQLS state is always UDA by its neighborhood RDMs \cite{FFQLS};
for later reference, we re-establish this result in the present context:
\begin{thm} [\cite{FFQLS}] 
\label{UDA}
If a multipartite pure state $\pure \in \H$ is DQLS relative to the neighborhood structure $\N = \{\N_j\}_{j=1}^M$, 
then $\pure$ is UDA by the set of all its neighborhood RDMs $\{ \rho_{\N_j} \}$.
\end{thm}

\begin{proof}
We first show that any quantum state $\tau \in \mathcal{D}(\H)$ with the same set of RDMs 
$\{\tr_{\overline{\N}_j}(\tau)=\rho_{\N_j}\}$ as $\pure \langle \psi |$, is such that $\supp(\tau) \subseteq \H_0(\pure)$.
Express $\tau \equiv \sum_i p_i |\phi_i\rangle \langle \phi_i|$ in its eigenbasis, and let 
$|\phi\rangle \equiv \sum_i \sqrt{p_i} |\phi_i\rangle|i\rangle \in \H \otimes \overline{\H}$
be a purification of $\tau$, with $\{|i\rangle\}$ forming an orthonormal basis in the ancilla space $\overline{\H}$. 
Also, let $|\overline{\psi}\rangle \equiv \pure |k\rangle$, with $|k\rangle \in \overline{\H}$.  
For fixed $j$, consider the 
$\N_j|\overline{\N}_j$ bi-partition of $\H$. By assumption,
$\tr_{\overline{\N}_j\overline{\H}} (|\phi\rangle\langle\phi |) = 
\tr_{\overline{\N}_j\overline{\H}}(|\overline{\psi}\rangle \langle \overline{\psi}|)$; it then follows that 
$(I_{\N_j}\otimes U_{\overline{ \N}_j\overline{ \H}}) |\overline{\psi}\rangle = |\phi\rangle$, for all $j$, where 
$U_{\overline{ \N}_j\overline{\H}} \in \mathcal{B}(\H_{\overline{ \N}_j}\otimes \H_{\overline{\H}})$ is a unitary 
transformation. Furthermore, each eigenvector of $\tau$ with non-zero eigenvalue is expressible as
$|\phi_i\rangle \equiv ( I_{\N_j} \otimes X^i_{\overline{ \N}_j})\pure$, with 
$X^i_{\overline{ \N}_j} = \langle i | U_{\overline{ \N}_j\overline{\H}}|k\rangle/ \sqrt{p_i},$ for all $p_i \neq 0$, and for 
all $\N_j \in \N$. As a result, $|\phi_i\rangle \in \H_0(\pure)$ for all $i$, by Theorem~\ref{DQLS:subspace}, which implies 
$\supp(\tau) \subseteq \H_0(\pure)$. Since $\pure$ is DQLS relative to $\N$, $\H_0(\pure)$ is one-dimensional 
and thus $\tau = \dpure$, as claimed. 
\end{proof}

\noindent 
In addition to showing that DQLS implies UDA, the above result can be useful in {\em searching for UDA states relative to} $\N$:
as the proofs shows, if the given pure state $\pure$ is not UDA, any state $\tau$, pure or mixed, that shares the same 
RDMs must be supported in the DQLS subspace of $\pure$ relative to $\N$. Hence, characterizing the DQLS subspace of 
a non-UDA pure state $\pure$ may be useful in reducing the search space for other global states that are ``joined'' or 
determined by the same RDMs as $\pure$. 

While DQLS implies UDA, the converse is {\em not} true in general: a counter-example is provided by the W-state 
on $N$-qubits, which is known to be UDA by the set of all its two-body RDMs \cite{W_UJ}, 
yet fails to be DQLS relative to any non-trivial neighborhood structure \cite{DQLS:p}.
That the DQLS property is a strictly stronger requirement than UDA is to be expected, since DQLS implies that 
knowledge about {\em the support of the  neighborhood RDMs alone} can fully specify the global state. 
Conversely, non-DQLS 
states that are UDA require complete information about the RDMs in order to uniquely specify the global state. 
Uniquely determined pure states are, in turn, a strict subset of all possible pure states: for instance, the 
GHZ-state on $N$-qubits, which is not DQLS relative to any non-trivial neighborhood structure \cite{DQLS:p}, 
is non-UDA (hence non-UDP) with respect to the set of $k$-body RDMs for any $k<N$.  
A pictorial summary is given in Fig. \ref{figure2}.

\begin{figure}[t]
\centering
\includegraphics[width=0.4\textwidth]{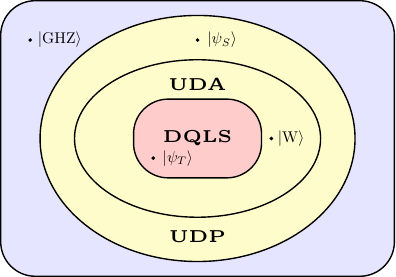}
\vspace*{1mm}
\caption{Inclusions between relevant subsets of multipartite pure quantum states. The two-excitation Dicke state 
on $N=4$ qubits, $|\psi_T\rangle$, was shown to be DQLS relative to any two three-body RDMs in \cite{DQLS:p}. 
An explicit symmetric state $|\psi_S \rangle$, also on $N=4$ qubits, that is UDP but not UDA was exhibited in \cite{UDAvsUDP}. }
\label{figure2}
\end{figure}

\subsection{Generic pure-state reconstruction in multipartite quantum systems}

Studying a multipartite quantum state based on the properties of its RDMs yields important insight on how correlations 
are distributed between different parties in the system: in fact, if an $N$-party state can be uniquely reconstructed 
using a given set of $k$-party RDMs ($k < N$), then the correlations it contains are limited to subsets of these $k$ 
parties \cite{Zhou}. Early results in this direction are due to Linden {\em et al.} \cite{3qubitUJ}, who established the UDA 
nature of generic pure states of 3-qubit systems with respect to their 2-party marginals. These studies were subsequently 
extended to arbitrary 3-qudit systems in \cite{red_states,red_states_impr,diosi} and to 4-qudit systems in \cite{gunhe}. 
In particular, in Linden \& Wootters showed that (i) a generic $\pure$ in $d_a \times d_b \times d_c$ dimensions 
(with $d_a \leq d_c$ w.l.o.g.), is UDA by $\rho_{ab}$ and $\rho_{bc}$, provided that $d_b \geq d_a + d_c -1$ and, building 
on this result, that (ii) a generic $\pure$ on $N$ qudits, with $N$ sufficiently large, is {\em UDA by only two} out 
of its $N \choose k$ $k$-body RDMs, where it suffices that $k \leq \lceil 2N/3\rceil$.
For $N=3$, it is worth noting that, assuming validity of Conjecture 3.11, the inequality in (i) implies $d_a d_b > d_c$, 
so that $\pure$ is indeed DQLS with respect to $\N_{\rm tri}$, hence consistently UDA.  For the case of $N$ qubits, 
the bound in (ii) was further improved to $k \leq \lceil N/2\rceil$ in \cite{red_states_impr}, at the cost of considering 
{\em about $N/2$} such RDMs.

Here, we strengthen the multipartite result by showing that a generic $N$-qudit pure state is UDA by only two of its 
$k$-body RDMs, with $k \approx N/2$, directly leveraging DQLS arguments. This affords two important advantages: 
not only does our approach avoid the need for full information about the RDMs, but it is {\em constructive} in nature, 
offering an explicit procedure for obtaining the global state via the DQLS condition in Eq.~\eqref{DQLS:eq}. 

\begin{prop}
\label{UDA:generic}

A generic multipartite pure state on $N$ qudits ($N > 3$) is UDA by its RDMs on any two neighborhoods $\N_{ab}$ and $\N_{bc}$ 
such that $\N_{ab} \cup \N_{bc} = \{1,2,\dots, N\}$, where it suffices that 
each contains no more than $ \floor{\frac{N+3}{2}}$ subsystems.
\end{prop}

\begin{proof} Let us organize the system of $N$ qudits into three subsystems $a,b$ and $c$, such that their dimensions 
are given by
\begin{itemize}
\item $d^{\frac{N-2}{2}} \times d^2 \times d^{\frac{N-2}{2}}$ for $N$ even, $d\geq 2$.
\item $d^{\frac{N-1}{2}} \times d \times d^{\frac{N-1}{2}}$ for $N$ odd, $d>2$.
\item $d^{\frac{N-3}{2}} \times d^2 \times d^{\frac{N-3}{2}+1}$ for $N$ odd, $d=2$.
\end{itemize}
From Theorem~\ref{special_cases}, we know that a generic tripartite pure state in $d_a \times d_b \times d_c$ dimensions 
($d_b >2$) is DQLS relative to the neighborhood structure $\N_{\rm tri} = \{\N_{ab},\N_{bc}\}$ when $d_a = d_c,d_b=d_c/2$. 
Hence, by Theorem~\ref{UDA}, in order to uniquely reconstruct a generic pure state in these dimensions, knowledge about the 
RDMs $\{\rho_{ab},\rho_{bc}\}$, is sufficient. The separate consideration of the $d=2$ case  (for $N$ odd) is due to the fact that 
generic pure states in $d_a \times 2 \times d_c$ dimensions are not DQLS when 
$d_a = d_c$, due to Theorem~\ref{qubitb}.
\end{proof}

It is important to appreciate that both the neighborhood structure and the {\em order} of the subsystems are, in general, 
important in determining the DQLS property of $\pure$. For example, a generic tripartite pure state in $2\times 2 \times 5$ 
dimensions is {\em not} DQLS relative to $\N_{\rm tri}= \{\N_{ab},\N_{bc}\}$, by Theorem~\ref{No-Go}. But, $\pure$ 
{\em is} DQLS relative to $\N_{\rm tri}' = \{\N_{ac},\N_{cb}\}$ (where $b$ and $c$ have been swapped), 
as generic states in $2 \times 5  \times 2$ dimensions are 
DQLS by Theorem~\ref{special_cases}. While $\N_{\rm tri}$ and $\N_{\rm tri}'$ are not equivalent from a control-theoretic 
standpoint as they represent different constraints, subsystems' ordering need not be important for 
quantum state reconstruction in typical settings where {\em all} $k$-body RDMs are assumed to be available.
Hence, when inferring the UDA property from DQLS results in such cases, it is allowable to permute subsystems as needed.  
With reference to the above example, choosing $\N_{\rm tri}'$ allows one to conclude that $\pure$ is UDA by all two-body 
RDMs -- in particular, $\{\rho_{ac},\rho_{bc}\}$ alone suffice.

One may verify that the results in Proposition~\ref{UDA:generic} extends those of \cite{diosi} which, 
while also being constructive, are {\em restricted to UDP}. Further to that, as anticipated, Eq.~\eqref{DQLS:eq} requires 
information only about the support of the relevant RDMs, as opposed to full information about these RDMs, as in previous 
works. We may thus summarize our result about generic-state reconstruction as follows: \emph{A generic pure state on 
$N$ qudits ($N$ large) is UDA by the support of any of its two $k$-body RDMs, where $k \approx N/2$ and the 
reconstruction procedure is given by Eq.~\eqref{DQLS:eq}}.

\section{Conclusion}
\label{sec:conclusion}

We have characterized the extent to which {\em generic} pure states on a finite-dimensional multipartite quantum system 
can be the unique asymptotic equilibrium state of purely dissipative Markovian dynamics, subject to specified 
locality constraints. Beside further contributing to the understanding of fundamental properties of both generic 
quantum states and quantum stabilization problems, we have shown how our results have implications for a 
number of problems relevant to QIP.  Specifically, on the one hand we have addressed QL stabilization of pure states 
drawn from a state $t$-design, as opposed to the full Haar measure, and approximate QL stabilization of pure states 
that are not stabilizable in exact form; on the other hand, we have revisited uniquely determined quantum states and 
quantum state tomography based on reduced density matrices in the light of the underlying stabilizability properties.

A number of related questions are prompted by the present analysis. First, it may be interesting to understand whether 
(and, if so, why) the fact that stabilizability occurs either with probability one or with probability zero is a feature exhibited 
by more general constrained stabilization settings than the QL tripartite one we have demonstrated here -- both in 
regard to $N>3$ and the possibility of different resource constraints (for instance, symmetry-constrained dissipative 
dynamics), as well as to in connection to ``phase transition'' phenomena that are known to occur for entanglement 
sharing and quantum marginals \cite{ss1}.  Second, by making contact with the 
analysis of mixed-state QL stabilization in \cite{FFQLS}, it is natural to ask about stabilizability properties of {\em 
random mixed states}.  In this case, one may consider the standard induced ensemble resulting from taking a partial 
trace over part of a random pure state or ensembles of random density matrices distributed according to other probability
measures \cite{bengtsson,Osipov}, again beginning from the simplest setting of a tripartite target system.

\section*{Acknowledgements}

It is a pleasure to thank Abhijeet Alase, Ramis Movassagh, and Michael Walter for insightful discussions on 
various aspects relevant to the present study. Work at Dartmouth was supported by the US NSF under Grant 
No. PHY-1620541, and the Constance and Walter Burke Special Projects Fund in Quantum Information Science.

\section{Appendix} 
\label{sec:app}

\subsection{Properties of generic matrices}
\label{sec:generic}

\begin{lemma}
\label{distinct_eig}
Let $A \in \mathds{M}^{d \times d}$ be a generic matrix. If $B$ is any matrix such that $[A,B] = \textbf{0}$, 
then $B = P(A)$, where $P(A)$ is a polynomial in $A$ with degree at most $(d-1)$.
\end{lemma}
\begin{proof}

The characteristic polynomial of the generic square matrix $A$ has arbitrary
coefficients, hence its roots are all distinct and $A$ is diagonalizable. Let $\{\mu_j(A)\}_{j=1}^d$ denote the spectrum of $A$ with $\mu_i(A) \neq \mu_j(A)$ when $i \neq j$. Since $[A,B] = \textbf{0}$, $A$ and $B$ are diagonal in the same basis. Let $\{\mu_j(B)\}_{j=1}^d$ be the spectrum of $B$ and $\vec{b} =[\mu_1(B),\dots, \mu_d(B)]^T$. In a basis where both $A$ and $B$ are diagonal, their non-zero entries are related as $\mu_j(b) = P(\mu_j(A))$, for all $j$, if $B$ is a polynomial in $A$ where $P(\cdot)$ is of some finite degree. In order to show this, construct the  Vandermonde matrix
\begin{equation*}
V_A \equiv \left[\begin{array}{c c c c} 1 & \mu_1(A) & \dots \mu_1^{d-1}(A)  \\ 1 & \mu_2(A) & \dots \mu_2^{d-1}(A) \\ \vdots & \ddots &\vdots \\ 1 & \mu_d(A) & \dots \mu_d^{d-1}(A) \end{array} \right].
\end{equation*}
Then, det$(V_A) = \prod_{1\leq i<j\leq d}(\mu_i(A)-\mu_j(A))$, which is non-zero as all the $\mu_j(A)$s are distinct. Accordingly, $V_A$ is invertible. Thanks to the Cayley-Hamilton theorem, any polynomial in $A$ is expressible as $P(A) = \sum_{k=0}^{d-1}c_kA^k$, $c_k\in \mathds{C}$. In a basis where $A$ (hence $P(A)$) is diagonal, the non-zero entries of $P(A)$ are given by the entries of the vector $V_A\vec{c}$, where $\vec{c} = [c_0,\dots, c_{d-1}]^T$. In order to verify that the same holds for $B$, choose in particular $\vec{c} = V_A^{-1}\vec{b}$, so that $\mu_j(B) = P(\mu_j(A))$ for all $j$; thus, $B = P(A)$, where $P(\cdot )$ is a polynomial of degree at most $(d-1)$, as claimed.
\end{proof}

\begin{lemma}
\label{generic_kernel}
Let $A \in \mathds{M}^{d \times d}$ be a generic matrix. $\vec{b} \in \mathds{C}^d$ is a generic vector independent of $A$. If $P(A)$ is a polynomial in $A$ satisfying $P(A)\vec{b} = \vec{0}$, then $P(A)$ is the zero polynomial.
\end{lemma}

\begin{proof} 
Since $A$ is a generic square matrix, it has $d$ distinct eigenvalues and is diagonalizable with its set of left/right eigenvectors forming a basis for $\mathds{C}^d$. Let $\{a_j\}_{j=1}^d$ denote the spectrum of $A$ and $\{\vec{\alpha}_j\}_{j=1}^d$ the corresponding right-eigenvectors, respectively, such that $A\vec{\alpha}_j = a_j \vec{\alpha}_j$, for all $j$. Let $P(A) = \sum_{k=0}^{d-1}p_kA^k$, for $p_k \in \mathds{C}$, for all $k$.  Owing to the generic nature of $\vec{b}$, $\vec{b} = \sum_{j=1}^d{b_j\vec{\alpha}_j}$, for $b_j \in \mathds{C}$ and $b_j \neq 0$, for all $j$. Following our assumption,  $\vec{b}$ belongs to the right kernel of $P(A)$. That is,
\begin{equation*}
P(A)\vec{b} = \sum_{j=1}^d{b_j(\sum_{k=0}^{d-1} p_k A^k)\vec{\alpha}_j} = \sum_{j=1}^d{b_j(\sum_{k=0}^{d-1} p_k a^k)\vec{\alpha}_j} = \vec{0}.
\end{equation*}
Since $\{\vec{\alpha}_j\}$ are a linearly independent set, it follows that $\sum_{k=0}^{d-1}{p_ka_j^k} = 0$ for all $j$, which in turn shows that $P(A) = \textbf{0}$.
\end{proof}

\begin{lemma}
\label{2commute}
Let $A,B \in \mathds{M}^{d \times d}$ be two generic matrices, independent of each other. The only non-trivial matrix that commutes with both of them simultaneously is the identity.
\end{lemma}

\begin{proof}
Let $X$ be a matrix that commutes simultaneously with $A$ and $B$. Due to the generic nature of $A$, it has all distinct eigenvalues and, following Lemma~\ref{distinct_eig}, $X = P(A)$, which is a polynomial of degree $(d-1)$ or lower. Also, $[P(A), B] = \textbf{0}$. Let $\{a_j\}_{j=1}^d$ denote the spectrum of $A$ and $\{\vec{\alpha}_j\}_{j=1}^d$ the corresponding right-eigenvectors, such that $A\vec{\alpha}_j = a_j \vec{\alpha}_j$, for all $j$. The latter form a basis for $\mathds{C}^d$.
Let $P(A) = \sum_{k=0}^{d-1} p_kA^k$, $p_k \in \mathds{C}$, for all $k$. If $\vec{\beta}$ denotes a right-eigenvector of $B$ with eigenvalue $b$, $BP(A)\vec{\beta} = bP(A)\vec{\beta}$. Due to the generic nature of $B$, all its eigenvalues are distinct and therefore, $P(A)\vec{\beta} = \lambda\vec{\beta}$ for some $\lambda \in \mathds{C}$. The eigenvectors of two independent generic matrices are also independent of each other, and hence $\vec{\beta} = \sum_{j=1}^d{c_j\vec{\alpha}_j}$, for $c_j \in \mathds{C}$ and $c_j \neq 0$, for all $j$. Hence,
\begin{equation*}
P(A)\vec{\beta} - \lambda\vec{\beta} = \sum_{j=1}^d{c_j \Big((\sum_{k=0}^{d-1}p_kA^k ) -\lambda \Big)\vec{\alpha}_j} = \sum_{j=1}^{d}{c_j \Big((\sum_{k=0}^{d-1} p_ka_j^k)-\lambda \Big)\vec{\alpha}_j} = \vec{0}.
\end{equation*}
This implies that $\sum_{k=1}^{d}{p_ka_j^k=\lambda}$ for all $j$, because $\{\vec{\alpha}_j\}$ are linearly independent. This in turn shows that $P(A) = \lambda I$, that is, $X$ is proportional to the identity.
\end{proof}

\subsection{Technical proofs}  
\label{sub:proofs}

We present here complete proofs that were not included in the main text for brevity.

\medskip

\noindent 
{\bf Theorem 3.7.} 
{\em Let $\pure$ be a generic tripartite pure state in 
$ \H_a \otimes \H_b \otimes \H_c$, with dimension $d_a \times d_b \times d_c$ and 
neighborhood structure $\N_{\rm tri} = \{\N_{ab},\N_{bc}\}$. When $d_b = 2$, 
\begin{enumerate}[i.]
\item $\pure$ is not DQLS relative to $\N_{\rm tri} $ when $d_a = d_c$. The dimension of $\H_0(\pure)$ is $d_a$.
\item $\pure$ is not DQLS relative to $\N_{\rm tri} $ when $d_c>d_a$, except for $d_c = d_a +1$. The dimension 
of $\H_0(\pure)$ is $\min\{d_a^2,(d_c-d_a)^2\}$.
\end{enumerate}}

\medskip
\noindent {\em Proof.}
We prove each statement separately. 
\begin{enumerate}[{\em i.}]

\item Let $d_a = d_c \equiv d$. The DQLS nature of generic pure states in $d\times 2 \times d$ dimensions is 
determined by the solution space of Eq.~\eqref{commute}. Since $A_{00}$ is a generic square matrix, it has all distinct 
eigenvalues. For this reason, any matrix that commutes with $A_{00}$ belongs to the space spanned by the first $(d-1)$ 
powers of $A_{00}$, including the identity (see Lemma~\ref{distinct_eig}). Thus, $\{X_a,X_c=X_a^T\}$ that satisfy 
Eqs.~\eqref{conditions} belongs to this $d$-dimensional space. According to Proposition~\ref{DQLS:generic}, the 
DQLS subspace $\H_0(\pure)$ is $d$-dimensional ($d>1$), hence $\pure$ is not DQLS.

\item In this case, $\bar{d} = d_c-d_a > 0$. We only need to focus on the case where $\bar{d} < d_a$, 
$\bar{d} \neq 0$, that is not covered under the no-go Theorem~\ref{No-Go}.  Let us reorganize 
Eqs.~\eqref{cqubit1}-\eqref{cqubit2} into a compact form by vectorizing the matrices $X_a,\tilde{X}_{0\bar{d}}$ 
and $\tilde{X}_{\bar{d}\bar{d}}$, i.e., by stacking up their columns and using the property 
${\rm vec}(AXB)= (B^T\otimes A){\rm vec}(X)$ \cite{bhatia}:
\begin{equation}
\label{vec}
\mathcal{C}\vec{v} \equiv
\left[\begin{array}{c c c}
A_{0\bar{d}}^T\otimes I_{00} & -I_{\bar{d}\bar{d}}\otimes A_{0\bar{d}} & \textbf{0} \\
A_{00}^T\otimes I_{00}-I_{00}\otimes A_{00}^T & \textbf{0} & -I_{00}\otimes A_{0\bar{d}}
\end{array} \right] 
\left[\begin{array}{l}
{\rm vec} (X_a)\\ {\rm vec}(\tilde{X}_{\bar{d}\bar{d}})\\ {\rm vec}( \tilde{X}_{\bar{d}0})
\end{array}\right]
= \left[\begin{array}{c}
\textbf{0}\\ \textbf{0}
\end{array} \right],
\end{equation}
Here, $I_{00},I_{\bar{d}\bar{d}}$ are the $d_a\times d_a$ and $\bar{d} \times \bar{d}$ identity matrices, respectively, 
$\mathcal{C}$ is called the coefficient matrix, and $ \vec{v}$ is formed by stacking the columns of ${\rm vec} (X_a), 
{\rm vec}(\tilde{X}_{\bar{d}\bar{d}})$ and  ${\rm vec}( \tilde{X}_{\bar{d}0})$.  This homogeneous system of linear equations 
has $\bar{d}d_a+d_a^2$ constraints and $\bar{d}d_a+d_a^2+\bar{d}^2$ unknowns. Evidently, it is under-constrained and 
has at least one non-trivial solution (corresponding to $X_a,\tilde{X}_{\bar{d}\bar{d}}$ proportional to the identity and 
$\tilde{X}_{0\bar{d}}$  being the zero matrix). The total number of linearly independent constraints in Eq.~\eqref{vec}, 
which is given by the row rank of $\mathcal{C}$, determines the dimension of the solution space. To determine the 
row rank of $\mathcal{C}$, observe that the first $d_a^2$ rows of $\mathcal{C}$ (upper block) are linearly independent 
of the rest (lower block), by construction. For the upper block, linearly dependent rows, if any, imply that there exist 
some non-zero vector ${\rm vec}(P)$ such that, 
${\rm vec}(P)^T[A_{0\bar{d}}^T\otimes I_{00} \; -I_{\bar{d}\bar{d}}\otimes A_{0\bar{d}}] = \vec{0}.$

This condition is equivalent to $P^TA_{0\bar{d}} = \textbf{0}$ and $A_{0\bar{d}}P^T = \textbf{0}$, where $P$ is the matrix form of ${\rm vec}(P)$. That is, columns (rows) of $P$ belong to the left (right) kernel of $A_{0\bar{d}}$. Due to its generic nature, $A_{0\bar{d}}$ does not have a non-trivial right kernel when $\bar{d} < d_a$. Therefore, there exists no non-trivial $P$ that simultaneously satisfies the conditions given above and hence, the rows in the upper block of $\mathcal{C}$ are linearly independent of each other. Similarly, linear dependencies of the lower block, if any, imply the existence of a vector ${\rm vec}(Q)$ such that
\begin{equation*}
{\rm vec}(Q)^T[A_{00}^T\otimes I_{00}-I_{00}\otimes A_{00}^T \quad -I_{00}\otimes A_{0\bar{d}}] = \vec{0}, 
\end{equation*}
or, equivalently, the conditions: $[A_{00},Q^T] = \textbf{0}$ and $Q^TA_{0\bar{d}} = \textbf{0}$, satisfied for some matrix $Q$. Any $Q^T$ that commutes with the generic matrix $A_{00}$ is such that $Q^T = P(A_{00})$, the latter being some polynomial in $A_{00}$, by Lemma~\ref{distinct_eig}. By invoking Lemma~\ref{generic_kernel}, it also follows that if $P(A_{00})A_{0\bar{d}}=\textbf{0}$ for another generic matrix $A_{0\bar{d}}$, then $P(A_{00})$ is the zero polynomial (trivial solution). This shows that the rows in the lower block of 
$\mathcal{C}$ are also linearly independent of each other.   Thus, we conclude that the coefficient matrix $\mathcal{C}$ has maximal row rank. Hence, the homogeneous system in Eq.~\eqref{vec} has $\bar{d}d_a+d_a^2$ independent constraints vs. 
$\bar{d}d_a+d_a^2+\bar{d}^2$ variables, and the dimension of the solution space is $\bar{d}^2$.

When $\bar{d} = 1$, it is thus verified that the only non-trivial solution corresponds to $X_a$ and $X_c$ proportional to the 
identity in the respective spaces. Accordinglyin $d_a \times 2 \times (d_a+1)$ dimensions, generic pure states are DQLS 
according to Proposition~\ref{DQLS:generic}. When $\bar{d} >1$ instead (with $\bar{d}<d_a$), Proposition~\ref{DQLS:generic} 
shows that the dimension of $\H_0(\pure)$, is $\bar{d}^2>1$, therefore generic states are non-DQLS, as claimed.\hfill$\Box$
\end{enumerate}

\newpage 

\noindent 
{\bf Theorem 3.10.} 
{\em Let $\pure$ be a generic tripartite pure state in $\H_a\otimes\H_b \otimes \H_c$, with neighborhood structure 
$\N_{\rm tri} = \{\N_{ab},\N_{bc}\}$ and central subsystem dimension $d_b >2$. Then, $\pure$ is DQLS relative 
to $\N_{\rm tri}$ for the following combinations of subsystem dimensions:  
\begin{enumerate}[i.]
\item $d_a \times d_b \times d_a$, 
\item $d_a \times d_b \times (d_a+1),$
\item $d_a \times d_b \times n d_a$, with $1<n<d_b$.
\end{enumerate}
}

\medskip
\noindent {\em Proof.}
We prove each statement separately.  
\begin{enumerate}[{\em i.}]

\item  $d_a \times d_b \times d_a$.
This case is an extension to Theorem~\ref{qubitb}(i), which shows that generic pure states in $d_a \times 2 \times d_a$ dimensions are not DQLS relative to $\N_{\rm tri}$. Eq.~\eqref{commute}, which determined the DQLS nature of generic pure states when $d_b = 2$, is now modified to $X_aA_{00}^i = A_{00}^iX_a,$ $i = 1, \dots, d_b-1.$
Due to Lemma~\ref{2commute}, the only non-trivial matrix that commutes with two or more generic matrices simultaneously is the identity. Thus, the only choice of $\{X_a,X_c=X_a^T\}$ that satisfy the set of equations~\eqref{conditions} for $d_b >2$ is proportional to the identity, and by Proposition~\ref{DQLS:generic}, generic pure states in  $d_a \times d_b \times d_a$ dimensions are DQLS relative to $\N_{\rm tri}$, for any $d_b>2$.

\item  $d_a \times d_b \times (d_a+1)$.
Following Theorem~\ref{qubitb}(ii), generic pure states in $d_a \times 2 \times (d_a+1)$ dimensions are DQLS relative to $\N_{\rm tri}$. It then suffices to apply Proposition~\ref{higher-d_b} to show that  generic pure states in $d_a \times d_b \times (d_a+1)$ dimensions are also DQLS relative to $\N_{\rm tri}$, for any $d_b >2$.

\item  $d_a \times d_b \times nd_a$, $1<n<d_b$.
When $n\geq d_b$, these cases fall under the no-go Theorem~\ref{No-Go}, and hence are not considered in this proof. 
We first prove that generic pure states in  $d_a \times 3 \times 2d_a$ dimensions are DQLS relative to $\N_{\rm tri}$. Consider the modified version of the set of equations~\eqref{conditions} given by Eqs.~\eqref{many1}-\eqref{many2}, for $d_b = 3$. The block decomposition of the relevant matrices are given in Table~\ref{block}. In this case, $A_{0\bar{d}}^i$ are generic square matrices since $\bar{d} = d_a$, and hence invertible. Eliminating $\tilde{X}_{\bar{d}\bar{d}}$ in  Eq.~\eqref{many1} for $i=1,2$ leads to the condition 
$ [ A_{0\bar{d}}^2(A_{0\bar{d}}^1)^{-1},X_a] = \textbf{0}.$
Let $Q \equiv  A_{0\bar{d}}^2(A_{0\bar{d}}^1)^{-1} $. Since $Q$ is a generic matrix, Lemma~\ref{distinct_eig} implies that $X_a = P(Q)$, with $P(\cdot)$ being a polynomial of degree at most $(d_a-1)$. Eq.~\eqref{many2} then rewrites
\begin{align*}
\tilde{X}_{0\bar{d}} &=  (A_{0\bar{d}}^1)^{-1}(P(Q)A^1_{00} - A^1_{00}P(Q)), 
\\
\tilde{X}_{0\bar{d}} &=  (A_{0\bar{d}}^2)^{-1}(P(Q)A^2_{00} - A^2_{00}P(Q)). 
\end{align*}
We see that RHS of both these equations are equal to each other. Owing to the full-rank nature of the generic matrices that are involved in these equations, we now show that they are satisfied only by $\tilde{X}_{0\bar{d}} = \textbf{0}$. For this purpose, observe that equating the LHS of the two equations we obtain $Q[P(Q),A^1_{00}] = [P(Q),A^2_{00}]$. This in turns leads to $[P(Q),(\bar{A}^1_{00}-A^2_{00})]=0$, where $\bar{A}^1_{00} = QA^1_{00}$. Due to the generic, invertible nature of the matrices involved in this relationship, $Q$ and $(\bar{A}^1_{00}-A^2_{00})$ are independent of each other. Following Lemma~\ref{2commute}, $P(Q) = I$, otherwise this implies that a non-identity matrix $P(Q)$ commutes with two independent, generic matrices $Q$ and $(\bar{A}^1_{00}-A^2_{00})$. Therefore, $X_a = P(Q) = I$. This in turn shows that  $\tilde{X}_{\bar{d}\bar{d}} = I$ and $\tilde{X}_{0\bar{d}} = \textbf{0}$.  Referring to Table~\ref{block}, these observations imply $X_c = I$. Thus, by Proposition~\ref{DQLS:generic}, generic pure states in $d_a \times 3 \times 2d_a$ dimensions are DQLS.

The above result is now extended with the help of Proposition~\ref{higher-d_b} and Proposition~\ref{partDQLS} to show that generic pure states in $d_a \times d_b \times nd_a$ dimensions are DQLS relative to $\N_{\rm tri}$, for any $d_b \geq 3$ and $1<n < d_b$. \hfill$\Box$
\end{enumerate}

\end{document}